\documentclass{ws-ijcga}
\usepackage{graphicx,amssymb,amsmath}
\usepackage{lmodern}

\usepackage{graphicx}
\graphicspath{{./images/}}
\usepackage{epsf}
\usepackage{epsfig}
\usepackage{color}
\usepackage{cite}
\usepackage[titlenumbered, ruled]{algorithm2e}
\usepackage[hidelinks]{hyperref}
\usepackage{bigfoot}
\usepackage{lineno}
\usepackage{todonotes}
\usepackage[font=small]{caption}
\usepackage{subfig}
\usepackage{microtype}

\newcommand{\qedhere}{%
  \begingroup \let\mathqed\math@qedhere
    \let\@elt\setQED@elt \QED@stack\relax\relax \endgroup
}

\def\comic#1#2#3{\parbox{#1}{\centering\includegraphics[width=#1]{#2}\\{\footnotesize #3}}}
\def\comicII#1#2{\parbox{#1}{\centering\includegraphics[width=#1]{#2}}}

\newcommand{\F}{\mathcal{F}}

\newtheorem{conj}{Conjecture}

\title{Folding Polyominoes into (Poly)Cubes\thanks{
A preliminary extended abstract appears in the Proceedings of the 27th Canadian Conference on Computational Geometry~\cite{abd+-fppc-15}.}}

\author{Oswin Aichholzer}
\address{Institute for Software Technology, Graz University of Technology,
\texttt{oaich@ist.tugraz.at}.}
\author{Michael Biro}
\address{Department of Mathematics and Statistics, Swarthmore College, {\tt mbiro1@swarthmore.edu}.}
\author{Erik D. Demaine, Martin L. Demaine}
\address{Computer Science and Artificial Intelligence Laboratory, Massachusetts Institute of Technology, \texttt{edemaine@mit.edu, mdemaine@mit.edu}.}
\author{David Eppstein}
\address{Computer Science Department, University of California, Irvine, \texttt{eppstein@ics.uci.edu}.}
\author{S\'{a}ndor~P.~Fekete}
\address{Department of Computer Science, TU Braunschweig,
Germany. \texttt{s.fekete@tu-bs.de}}
\author{Adam Hesterberg}
\address{Department of Mathematics, 
Massachusetts Institute of Technology, \texttt{achester@mit.edu}}
\author{Irina Kostitsyna}
\address{Computer Science Department, 
Universit\'{e} Libre de Bruxelles, \texttt{irina.kostitsyna@ulb.ac.be}.}
\author{Christiane Schmidt}
\address{Communications and Transport Systems, ITN, 
Link\"oping University, {\tt christiane.schmidt@liu.se}. }

\index{Aichholzer, Oswin}
\index{Biro, Michael}
\index{Demaine, Erik}
\index{Demaine, Martin}
\index{Eppstein, David}
\index{Fekete, S\'{a}ndor}
\index{Hesterberg, Adam}
\index{Kostitsyna, Irina}
\index{Schmidt, Christiane}

\begin{document}
\thispagestyle{empty}
\maketitle

\begin{abstract}
We study the problem of folding a polyomino $P$ into a polycube~$Q$,
allowing faces of $Q$ to be covered multiple times.
First, we define a variety of folding models according to whether the folds
(a)~must be along grid lines of $P$ or can divide squares in half
(diagonally and/or orthogonally),
(b)~must be mountain or can be both mountain and valley,
(c)~can remain flat (forming an angle of $180^\circ$), and
(d)~must lie on just the polycube surface or can
have interior faces as well.
Second, we give all the inclusion relations among all models that fold on the
grid lines of~$P$.
Third, we characterize all polyominoes that can fold into a unit cube,
in some models.
Fourth, we give a linear-time dynamic programming algorithm to fold a
tree-shaped polyomino into a constant-size polycube, in some models.
Finally, we consider the triangular version of the problem,
characterizing which polyiamonds fold into a regular tetrahedron.
\end{abstract}

\section{Introduction}\label{sec:intro}

When can a polyomino $P$ be folded into a polycube $Q$?
The answer to this basic question depends on the precise model of
folding allowed.
At the top level, two main types of foldings
have been considered in the literature: \emph{polyhedron folding}
(as in Part~III of \cite{Demaine-O'Rourke-2007}) and \emph{origami folding}
(as in Part~II of \cite{Demaine-O'Rourke-2007}).

Polyhedron folding perhaps originates with \cite{Lubiw-O'Rourke-1996}.
In this variant of the problem,
the folding of $P$ must exactly cover the surface
of~$Q$, with a single layer everywhere.  This requirement is motivated by
sheet metal bending and other manufacturing applications with thick material,
where it is desired to form a given shape with the material without creating layers of doubled thickness.
If we make the additional assumption that at least one square of $P$ maps to a whole square of $Q$,
an easy polynomial-time algorithm can determine whether $P$ folds into
$Q$ in this model: guess this mapping (among the $O(|P| \cdot |Q|)$ choices),
and then follow the forced folding to test whether $Q$ is covered singly
everywhere.  Without this assumption, algorithms for edge-to-edge gluing
of polygons into convex polyhedra
\cite{Lubiw-O'Rourke-1996}, \cite[ch.~25]{Demaine-O'Rourke-2007}
can handle the case where the grid of $P$ does not match the grid of~$Q$,
but only when $Q$ is convex.
Specific to polyominoes and polycubes, there is extensive work in this model
on finding polyominoes that fold into many different polycubes
\cite{AloBosCol-CCGA-10} and into multiple different boxes
\cite{AbeDemDem-CCCG-11,Mitani-Uehara-2008,Shirakawa-Uehara-2013,Uehara-box-survey,Xu2015}

In origami folding, the folding of $P$ is allowed to multiply cover the
surface of $Q$, as long  every point of $Q$ is covered by at least one layer of~$P$.
This model is motivated by folding thinner material such as paper,
where multiple layers of material are harmless,
and is the model considered here.
General origami folding results \cite{Demaine-Demaine-Mitchell-2000}
show that every polycube, and indeed every polyhedron, can be folded from
a polyomino of sufficiently large area.
Two additional results consider the specific case of folding a polycube,
with precise bounds on the required area of paper.
First, every polycube $Q$ can be folded from a sufficiently large $\Omega(|Q|) \times \Omega(|Q|)$ square polyomino
\cite{Benbernou-Demaine-Demaine-Ovadya-2009}.
(This result also folds all the interior grid faces of~$Q$.)
Second, every polycube $Q$ can be folded from a
$2 |Q| \times 1$ rectangular polyomino,
with at most two layers at any point
\cite{Benbernou-Demaine-Demaine-Lubiw-2017}.
Both of these universality results place all creases on a fixed
\emph{hinge pattern} (typically a square grid plus some diagonals),
so that the target polycube $Q$ can be specified
just by changing the fold angles in a polyomino $P$
whose shape depends only on the size~$|Q|$ and not on the shape of~$Q$.

\begin{figure}
  \centering
  \includegraphics{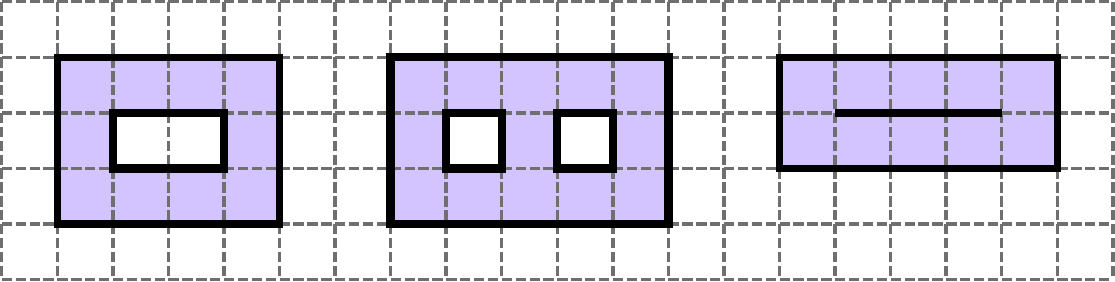}
  \caption{Three polyominoes that fold along the grid lines into a unit cube,
    from puzzles by Nikolai Beluhov \cite{Beluhov-2014}.}
  \label{puzzles}
\end{figure}

In this paper, we focus on the decision question of whether a given
polyomino~$P$ can fold into a given polycube~$Q$, in the origami folding
model.  In particular, we characterize the case when $Q$ is a single cube,
a problem motivated by three puzzles by Nikolai Beluhov \cite{Beluhov-2014};
see Figure~\ref{puzzles}.
This seemingly simple problem already turns out to be surprisingly intricate.

We start in Section~\ref{sec:not}
by refining different folding models that allow or disallow various types of folds:
(a)~folds being just along grid lines of $P$ or also along diagonal/orthogonal
bisectors of grid squares,
(b)~folds being all mountains or a combination of mountain and valley,
(c)~folds being strictly folded or can remain flat
(forming an angle of $180^\circ$), and
(d)~ folds being restricted to lie on just the polycube surface or folds being allowed to
have interior faces as well.
Then, in Section~\ref{sec:foldhier}, we characterize all grid-aligned models
into a hierarchy of the models' relative power.

Next, in Section~\ref{sec:char}, we characterize all the polyominoes that
can be folded into a unit cube, in grid-based models.
In particular, we show that \emph{all} polyominoes of at least ten unit squares
can fold into a unit cube using grid and diagonal folds, and list all smaller
polyominoes that cannot fold into a unit cube in this model.
Observe that the polyominoes of Figure~\ref{puzzles} each consist of at least
ten unit squares, though they furthermore fold into a cube using just
grid-aligned folds.  In this more strict model, we characterize which
tree-shaped polyominoes lying in a $2 \times n$ and $3 \times n$ strip fold into a unit cube.
This case matches the puzzles of Figure~\ref{puzzles} except in the
``tree-shaped'' requirement; it remains open to
characterize cyclic polyominoes that fold grid-aligned into a unit cube.

We also give some algorithmic results and consider the triangular grid.
In Section~\ref{sec:dp-tree}, we give a linear-time dynamic programming
algorithm to fold a tree-shaped polyomino into a constant-size polycube,
in grid-based models.
In Section~\ref{sec:triang}, we consider the analogous problems
on the triangular grid, characterizing (rather easily) which polyiamonds fold
into a regular tetrahedron, in grid-based models.
Finally, we present some open problems in Section~\ref{sec:con}.

\section{Notation}\label{sec:not}
A \emph{polyomino} $P$ is a two-dimensional polygon formed by a union of  $|P|=n$ \emph{unit
squares} on the square lattice connected edge-to-edge. We do not require all adjacent pairs of squares to be connected to each other; that is, we allow ``cuts'' along the edges of the lattice.
A polyomino is a \emph{tree shape} if the dual graph of its unit squares (the graph with a vertex per square and an edge for each connected pair of squares) is a tree.  Analogously, in one dimension higher, a \emph{polycube}
is a connected three-dimensional polyhedron formed by a union of \emph{unit cubes} on the
cubic lattice connected face-to-face.
If a polyomino or polycube is a rectangular box,
we refer to it by its exterior dimensions, e.g., a
2$\times$1 polyomino (domino) or a 2$\times$2$\times$1 polycube. We denote the special case of a single unit cube by $C$.

We study the problem of folding a given polyomino $P$ to form a given polycube~$Q$, allowing various different combinations of types of folds: \emph{axis-aligned}
$+90^{\circ}$ and $+180^{\circ}$ mountain folds, $-90^{\circ}$ and $-180^{\circ}$ valley folds, 
fold angles of \emph{any} degree, \emph{diagonal folds} through opposite corners of a square, and
\emph{half-grid} folds that bisect a unit square in an axis-parallel fashion.%
\footnote{We measure fold angles as dihedral angles between the two incident
  faces, signed positively for mountain (convex) and negatively for valley
  (reflex) folds.}

\begin{figure}[b]
\centering
\subfloat[]{\label{fig:int-f-a}
\includegraphics[scale=0.5]{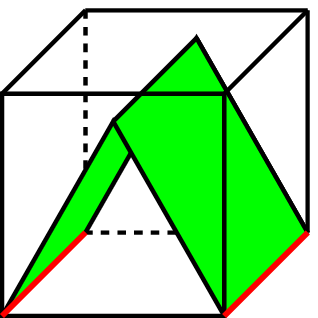}}
\hfil
\subfloat[]{\label{fig:int-f-b}
\includegraphics[scale=0.5]{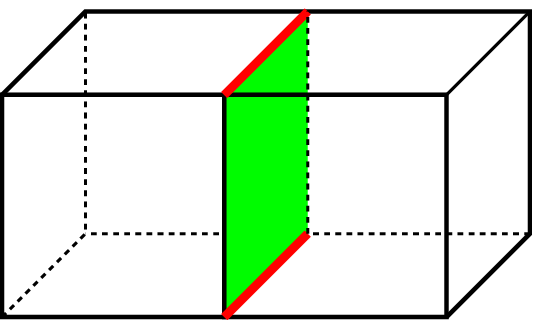}}
\caption{\label{fig:int-f} Examples for interior faces shown in green. They connect to the red cube edges.}
\end{figure}

All faces of $Q$ must be entirely covered by folded faces of $P$, but in some cases we allow some faces of $P$ to lie inside $Q$ without covering any face of $Q$.
A face of $P$ is an \emph{interior face} of $Q$ if it is not folded onto any of the (boundary) faces of~$Q$.
For example, if we consider a chain of unit squares in a polyomino, two folds of $+60^{\circ}$ along edges of two opposite edges of a unit square of $P$ result in a V-shape with two interior faces; 
see Figure~\ref{fig:int-f-a} and~\subref{fig:int-f-b} for examples.
A folding model $\mathcal F$ specifies a subset of $\F = \{+90^{\circ}$,$ -90^{\circ}$, $+180^{\circ}$, $-180^{\circ}$, any$^{\circ}$; grid; interior faces; diagonal; half-grid$\}$
as allowable folds. 


\section{Folding hierarchy}\label{sec:foldhier}


We say that model $\F_{x}$ is stronger than $\F_{y}$ (denoting this relation by $\F_{x}\ge \F_{y}$) if,
for all polyomino-polycube pairs $(P,Q)$ such that $P$ folds
into $Q$ in $\F_{y}$, $P$ also folds into $Q$ in $\F_{x}$. If there also exists a pair $(P',Q')$ such that 
$P'$ folds into $Q'$ in $\F_{x}$, but not in $\F_{y}$, then $\F_{x}$ is strictly stronger than $\F_{y}$ ($\F_{x}> \F_{y}$).
The relation `$\ge$' satisfies the properties of reflexivity, transitivity and antisymmetry, and therefore it defines a partial order on the set of folding models.
Figure~\ref{fig:hierarchy} shows the resulting hierarchy of the folding models that consist of combinations of the following
folds: $\{+90^{\circ}$,$ -90^{\circ}$, $+180^{\circ}$, $-180^{\circ}$,
any$^{\circ}$; interior faces; grid$\}$. 

\begin{figure}[h!]
\centering
    \includegraphics{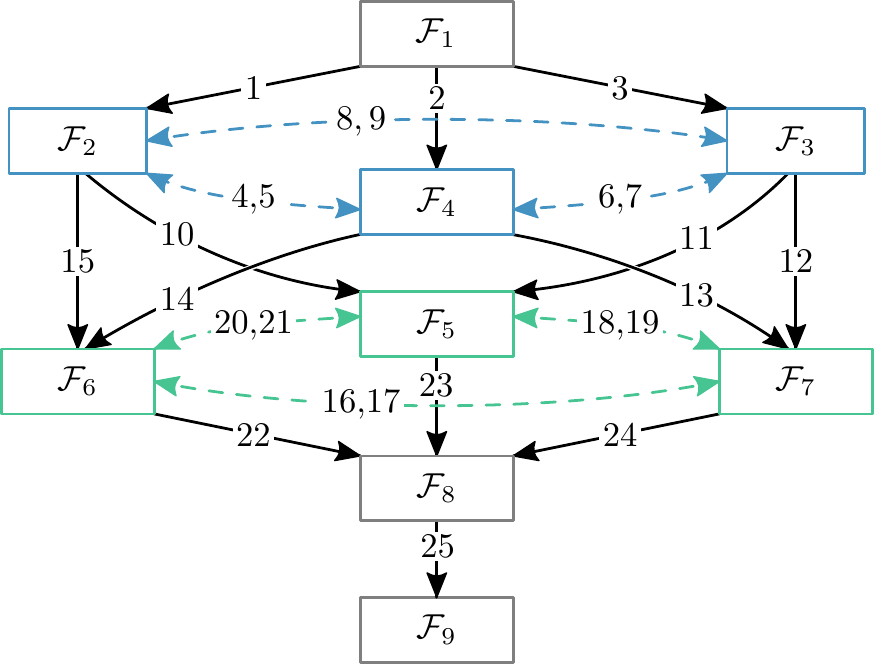}
         \caption{\label{fig:hierarchy} Hierarchy of fold operations. The nodes correspond to the following folding models: $\F_{1}=\{+90^{\circ}$; grid$\}$, $\F_{2}=\{\pm 90^{\circ};$ grid$\}$, $\F_{3}=\{+90^{\circ};$ interior faces; grid$\}$, $\F_{4}=\{+90^{\circ},180^{\circ};$ grid$\}$, $\F_{5}=\{\pm 90^{\circ};$ interior faces; grid$\}$, $\F_{6}=\{\pm 90^{\circ},180^{\circ};$ grid$\}$, $\F_{7}=\{+90^{\circ},180^{\circ};$ interior faces; grid$\}$, $\F_{8}=\{\pm 90^{\circ},180^{\circ}$; interior faces; grid$\}$, and $ \F_{9}=\{\text{any}^{\circ};$ interior faces; grid$\}$. A black arrow from $\F_y$ to $\F_x$ indicates that $\F_x>\F_y$. Blue and green arrows indicate incomparable models.
        The labels on arrows correspond to the relations' numbers from the proof of Theorem~\ref{le:hier}.
         }
\end{figure}

\begin{figure}
\centering
    \includegraphics[width=0.4\textwidth]{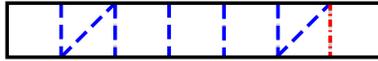}
         \caption{\label{fig:1x7} A 1$\times$7 
polyomino can be folded into a unit cube $C$ in model $\{\pm 90^{\circ},180^{\circ}$; grid; 
diagonal$\}$. The mountain and valley folds are shown in red and blue, respectively.
         }
\end{figure}

\begin{figure}
\centering
\subfloat[]{\label{fig:lb9-a}
\includegraphics[width=.225\textwidth]{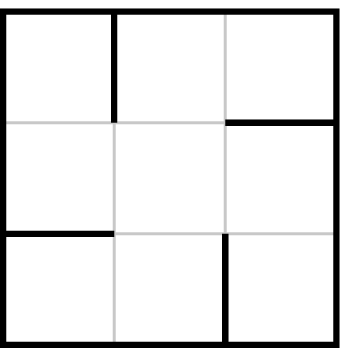}}
\hfil
\subfloat[]{\label{fig:lb9-b}
\includegraphics[width=.225\textwidth]{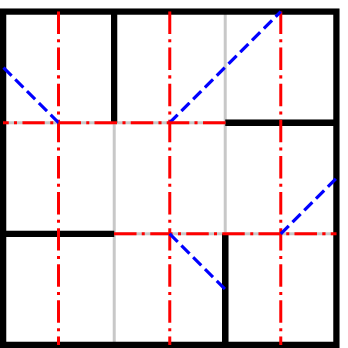}}
\caption{\label{fig:lb9}  In case we allow $\pm 90^{\circ},180^{\circ}$, half-grid folds, and diagonal folds only, shape $P$ does not fold to a unit cube $C$ . However, if we do not require faces of $C$ to be covered by full unit squares of $P$, $P$ does fold into $C$, the mountain (red) and valley folds (blue) are shown in \protect\subref{fig:lb9-b}.}
\end{figure}

Integrating diagonal and half-grid folds (which we omit from the hierarchy of Figure~\ref{fig:hierarchy})
can result in even stronger models. For instance, a 1$\times$7 
polyomino can be folded into a unit cube $C$ in model $\{\pm 90^{\circ},180^{\circ}$; grid; 
diagonal$\}$ (see Figure~\ref{fig:1x7}), but not in $\{$any$^{\circ}\text{; interior faces; grid}\}$ (the strongest model shown in Figure~\ref{fig:hierarchy}). The example shown in
Figure~\ref{fig:lb9} shows that $\F_{all} = \{$any$^{\circ}$;
interior faces; grid; diagonal; half-grid$\}$ is strictly stronger than $\F = \{$
any$^{\circ}$; interior faces; grid; diagonal$\}$.

We need the following lemma to prove the hierarchy of Figure~\ref{fig:hierarchy}.
\begin{figure}
\centering
\includegraphics[height=0.2\textwidth]{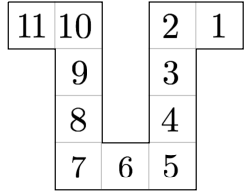}
\caption{\label{fig:int} The shape $P_2$ needs interior faces to be folded into a $2\times1\times1$ polycube.}
\end{figure}

\begin{figure}
    \comic{.23\textwidth}{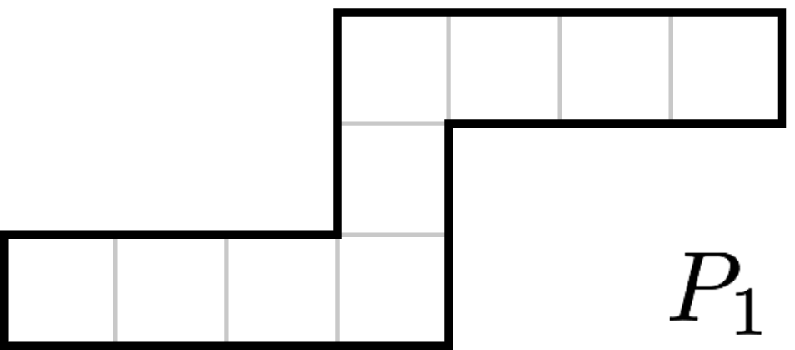}{\vspace*{3.4cm}(a)}
\hfill
\comic{.17\textwidth}{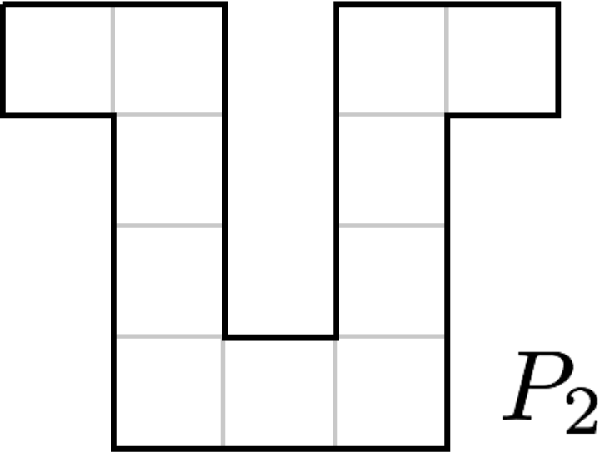}{\vspace*{3.1cm}(b)}
\hfill
\comic{.4\textwidth}{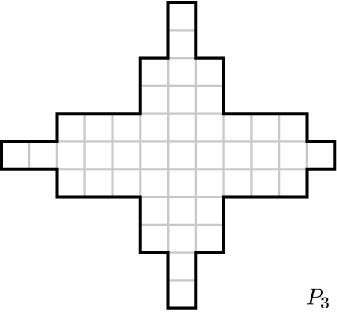}{(c)}
\\
\center
\vspace*{.5cm}
\comic{.6\textwidth}{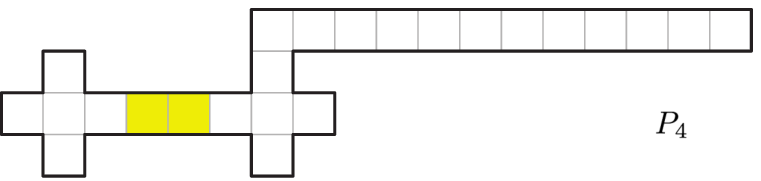}{(d)}\\\vspace*{.8cm}
\comic{.12\textwidth}{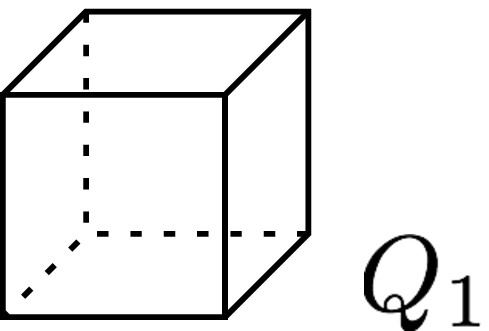}{\vspace*{1.5cm}(e)}\hfill
\comic{.192\textwidth}{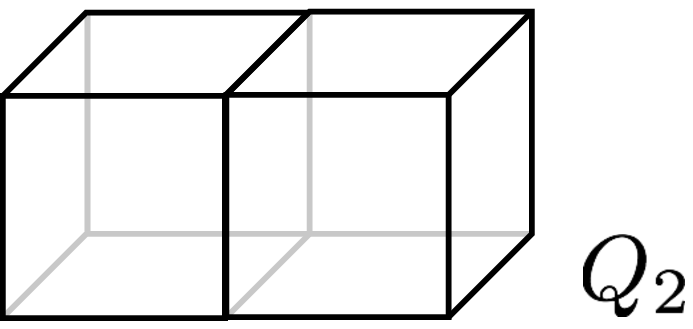}{\vspace*{1.4cm}(f)}\hfill
\comic{.264\textwidth}{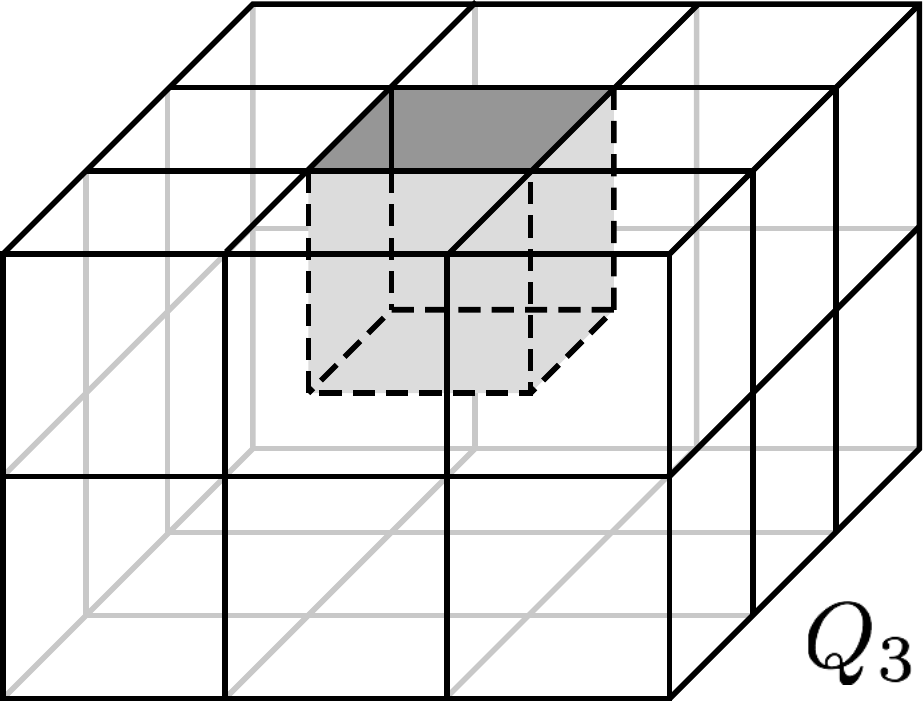}{(g)}\hfill
\comic{.264\textwidth}{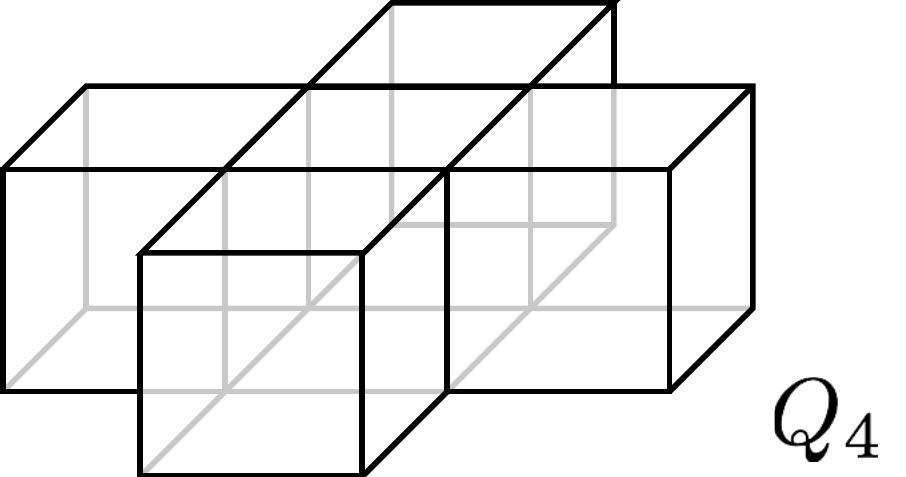}{\vspace*{.75cm}(h)}
         \caption{\label{fig:hierex} Examples for Theorem.~\ref{le:hier}: polyominoes (a) $P_1$, (b) $P_2$, (c) $P_3$, and (d) $P_4$; and polycubes (e) a unit cube $Q_1=C$, (f) a 2$\times$1$\times$1 polycube $Q_2$, (g) a 3$\times$3$\times$2 polycube $Q_3$ with a 1-cube hole centered in a 3$\times$3 face, and (h) a $5$-cube cross $Q_4$.
}
\end{figure} 

\begin{lemma}\label{le:b-int}
Folding the tree shape $P_2$ shown in Figure~\ref{fig:hierex}(b) into the $2\times1\times1$ polycube $Q_2$ shown in Figure~\ref{fig:hierex}(f) with 
$\F_{8}=\{\pm 90^{\circ},180^{\circ}\text{; interior faces; grid}\}$ requires interior faces. (Either of the faces $3$ or $9$ shown in Figure~\ref{fig:int} can be the interior face.)
\end{lemma}

\begin{proof}
We label the faces of $P_2$ as shown in Figure~\ref{fig:int}. The case analysis below shows that a folding without interior faces does not exist, as in every case we are forced to fold away at least two faces. Since $|P_2|=11$ and $Q_2$ has 10 faces, there will not be enough remaining faces to complete the folding. 

First note that any $\pm 180^\circ$ folds will reduce the figure to at most $10$ squares, and it is easy to check that these cannot be folded to $Q_2$ without an additional overlap. Furthermore, all of the folds between pairs of exterior faces of $Q_2$ have the same orientation, so (in the absence of $\pm 180$ degree folds to reverse the orientation) any fold from $P_2$ to $Q_2$ that avoids interior faces can only use one of the two angles  $+90^\circ$ or $-90^\circ$. Therefore, it is sufficient to check that $P_2$ cannot be folded to $Q_2$ using only $+90^\circ$ folds and only exterior faces.

If face $5$ covers one of the 1$\times$1 faces of $Q_2$, faces $9$ and $11$ are forced to cover faces $1$ and $2$, respectively, thus doubling two faces. If face $5$ covers part of a $1 \times 2$ face of $Q_2$, then there are two subcases: If face $4$ or face $6$ covers the adjacent $1 \times 1$ face of $Q_2$ then faces $7$ and $8$ are forced to cover faces $3$ and $4$, respectively, thus doubling two faces. If face $4$ or face $6$ covers the other half of the $1 \times 2$ face of $Q_2$ then faces $9$, $10$ and $11$ are forced to collectively cover faces $1$, $3$ and $4$, thus doubling at least two faces. Therefore, it is not possible to fold $P$ into $Q_2$ without interior faces.

\end{proof}

The following establishes the relationships between models presented in Figure~\ref{fig:hierarchy}.


\begin{theorem}\label{le:hier}
The folding models consisting of combinations of the following folds \{$+90^{\circ}$, $-90^{\circ}$, $+180^{\circ}$, $-180^{\circ}$, any$^{\circ}$; 
interior faces; grid\} have the mutual relations presented in Figure~\ref{fig:hierarchy}. In particular, these mutual relations hold for polyominoes without holes.
\end{theorem}
\begin{proof}
First note that all containments are obvious: for any two models $\F_x$ and $\F_y$ corresponding to two nodes of the hierarchy given in Figure~\ref{fig:hierarchy} such that there is an arrow from $\F_y$ to $\F_x$, trivially $\F_x \geq \F_y$. It remains to give examples to show that either two models are incomparable or that one is strictly stronger than the other, that is, to show that there exists a polyomino that folds into a given polycube under one folding model, but not under the other.
All cases can be shown using the 
polyomino-polycube pairs 
shown in Figure~\ref{fig:hierex}. They fold with: (a) $+90^{\circ}$, $180^{\circ}$ folds, (b) $+90^{\circ}$ and interior faces (see Lemma~\ref{le:b-int}), (c) $\pm90^{\circ}$folds, (d) $+90^{\circ}$, $180^{\circ}$ and interior faces.


Example (a) establishes relations $2$, $4$, $7$, $12$, $15$, $19$, $20$, $23$.

Example (b) establishes relations $3$, $6$, $9$, $10$, $13$, $16$, $21$, $22$.

Example (c) establishes relations $1$, $5$, $8$, $11$, $14$, $17$, $18$, $24$.

Example (d) establishes relation $25$.

Let $\F_{1}=\{+90^{\circ}\text{; grid}\}$, $\F_{2}=\{\pm 90^{\circ}\text{; grid}\}$, $\F_{3}=\{+90^{\circ}\text{; interior faces; grid}\}$, $\F_{4}=\{+90^{\circ}, 180^{\circ}\text{; grid}\}$, $\F_{5}=\{\pm 90^{\circ}\text{; interior faces;}\text{ grid}\}$, $\F_{6}=\{\pm 90^{\circ}, 180^{\circ}\text{; grid}\}$, $\F_{7}=\{+90^{\circ}, 180^{\circ}\text{; interior faces; grid}\}$, $ \F_{8}=\{\pm 90^{\circ}, 180^{\circ}\text{; interior faces; grid}\}$, and $ \F_{9}=\{\text{any}^{\circ}\text{; interior faces; grid}\}$.

First we will give a detailed proof for {\bf Relation 25}. We will then present the proofs for the other relations more briefly.

{\bf Relation 25}: $\F_{9}=\{\text{any}^{\circ}\text{; interior faces; grid}\}$ is strictly stronger than $\F_{8}=\{\pm 90^{\circ}, 180^{\circ}\text{; interior faces; grid}\}$. Any polyomino $P$ that folds into polycube $Q$ in~$\F_{8}$ also folds into $Q$ in $\F_{9}$. To prove a strict relation, we must show that there exists some polyomino $P'$ that folds into some polycube $Q'$ in $\F_{9}$, but that does not fold into $Q'$ in $\F_{8}$.
Let $Q'$ consist of five cubes forming a cross as in Figure~\ref{fig:hierex}(h) ($Q'=Q_4$), and let $P'$ be as in Figure~\ref{fig:hierex}(d) ($P'=P_4$). Assume that $P'$ can be folded into $Q'$ in the folding model~$\F_{8}$. 
Polyomino $P'$ consists of $24$ unit squares, while $Q'$ has $22$ square faces on its surface. Therefore, $22$ out of the $24$ squares of $P'$ will be the faces of~$Q'$ when folded. Consider the 12$\times$1 subpolyomino of $P'$. When folded, it has to form the ``walls'' of the cross $Q'$, that is, the vertical faces of $Q'$, otherwise this strip can form not more than eight faces of the cross (looping around one of the 3$\times$1$\times$1 subpolycubes). It is straightforward to see now, that in $\F_{8}$ the two yellow squares prevent $P'$ from folding into $Q'$. However, in $\F_{9}$ the two yellow squares can form interior faces with a $60^{\circ}$ interior fold. Therefore, $\F_{8}<\F_{9}$.

{\bf Relation 1:} We show that $P_3$ folds into $Q_3$ in~$\F_{2}$, but does not fold into $Q_3$ in~$\F_{1}$. $|P_3|=46$, and $Q_3$ has $46$ square faces on its surface. Hence, all squares of $P_3$ are faces of $Q_3$. Consider the $3\times 9$ subpolyomino, $P_{3\times 9}$, of $P_3$. When folded with  $+90^{\circ}$ only, there exist only two positions for this subpolyomino in $Q_3$. In both positions the center $3\times 3$ piece covers $Q_3$'s $3\times 3$ face; the two positions differ only by a rotation of $90^{\circ}$. But then the two flaps of size one and two, respectively, adjacent to $P_{3\times 9}$ need to fold into the 1-cube hole, using $90^{\circ}$ folds at the flap grid lines adjacent to $P_{3\times 9}$. But a $90^{\circ}$ fold at the remaining grid line of the size two flap does not result in the outer square covering a separate face of $Q_3$. This contradicts the assumption that all squares of $P_3$ are faces of~$Q_3$.

{\bf Relation 2:} We show that $P_1$ folds into $Q_1$ in~$\F_{4}$, but does not fold into $Q_1$ in $\F_{1}$. $|P_1|=9$, while $Q_1$ has $6$ square faces on its surface. Using $+90^{\circ}$ folds the two $1\times 4$ subpolyominoes of $P_1$ fold to two rings of size four. Folding along both grid lines adjacent to the connecting square with $+90^{\circ}$ those two rings overlap. This results in an overlap of four faces, making it impossible to constitute a shape of $6$ faces with the $9$-square shape $P_1$. In $\F_{4}$ one grid line adjacent to the connecting square is folded with $180^{\circ}$, the other with $+90^{\circ}$. Folding all other grid lines with~$+90^{\circ}$ results in $Q_1$.

{\bf Relation 3:} We show that $P_2$ folds into $Q_2$ in $\F_{3}$, but does not fold into $Q_2$ in~$\F_{1}$. $|P_2|=11$, while $Q_2$ has $10$ square faces on its surface. Hence, only one square of $P_2$ can either overlap with another square or constitute an interior face of $Q_2$.
$\F_{1}$ does not allow interior faces; hence, Lemma~\ref{le:b-int} shows that $P_2$ does not fold into $Q_2$ in $\F_{1}$.

{\bf Relation 4:} We show that $P_1$ folds into $Q_1$ in $\F_{4}$ (see proof of Relation 2), but does not fold into $Q_1$ in $\F_{2}$. $|P_1|=9$, while $Q_1$ has $6$ square faces on its surface. To cover $Q_1$ at least one of the $1\times 4$ subpolyominoes of $P_1$ must be folded into a ring of size four using either $+90^{\circ}$ folds only or $-90^{\circ}$ folds only; without loss of generality we assume that the fold uses only $+90^{\circ}$ folds. If we fold the grid line of the connecting square adjacent to this ring with $-90^{\circ}$, the connecting square does not cover a face of $Q_1$, thus we need to use a $+90^{\circ}$ fold. The same holds true for the other grid line adjacent to the connecting square, and for the remaining grid lines. Then the proof of Relation 2 shows that we cannot fold $P_1$ into $Q_1$ using $+90^{\circ}$ folds only.

{\bf Relation 5:} We show that $P_3$ folds into $Q_3$ in $\F_{2}$ (see proof of Relation 1), but does not fold into $Q_3$ in $\F_{4}$. $|P_3|=46$, and $Q_3$ has $46$ square faces on its surface. Hence, all squares of $P_3$ are faces of $Q_3$. Again, there only exist two positions for~$P_{3\times 9}$. The two flaps of size one and two adjacent to $P_{3\times 9}$ again need to constitute faces of the 1-cube hole. The flap grid lines adjacent to $P_{3\times 9}$ need to be folded with $+90^{\circ}$ folds, but then neither a $+90^{\circ}$ fold, nor a $180^{\circ}$ fold at the remaining grid line of the size two flap, result in the outer square covering a separate face of $Q_3$.

{\bf Relation 6:} We show that $P_2$ folds into $Q_2$ in $\F_{3}$ (see proof of Relation 3), but does not fold into $Q_2$ in $\F_{4}$. No interior faces are allowed in $\F_{4}$; hence, Lemma~\ref{le:b-int} shows that $P_2$ does not fold into $Q_2$ in $\F_{4}$.

{\bf Relation 7:} We show that $P_1$ folds into $Q_1$ in $\F_{4}$ (see proof of Relation 2), but does not fold into $Q_1$ in $\F_{3}$. $|P_1|=9$, while $Q_1$ has $6$ square faces on its surface. The proof of Relation 2 showed that using only $+90^{\circ}$ folds is insufficient to fold $P_1$ into~$Q_1$. That is, we would have to use interior faces to fold $P_1$ into $Q_1$ in $\F_{3}$. As~$Q_1$ is a unit cube, we need two folds of $+60^{\circ}$, an infeasible grid fold in $\F_{3}$. 

{\bf Relation 8:} We show that $P_3$ folds into $Q_3$ in $\F_{2}$ (see proof of Relation 1), but does not fold into $Q_3$ in $\F_{3}$. $|P_3|=46$, and $Q_3$ has $46$ square faces on its surface. Hence, all squares of $P_3$ are faces of $Q_3$. In the proof of Relation 1 we showed that $P_3$ does not fold into $Q_3$ in $\F_{1}=\{+90^{\circ}\text{; grid}\}$. As $\F_{3}=\{+90^{\circ}\text{; interior faces; grid}\},$ we need to use interior faces to fold $P_3$ into $Q_3$. But any such interior face would be a square of $P_3$ that is not a face of $Q_3$.

{\bf Relation 9:} We show that $P_2$ folds into $Q_2$ in $\F_{3}$ (see proof of Relation 3), but does not fold into $Q_2$ in $\F_{2}$. No interior faces are allowed in $\F_{2}$; hence, Lemma~\ref{le:b-int} shows that $P_2$ does not fold into $Q_2$ in $\F_{2}$.

{\bf Relation 10:} We show that $P_2$ folds into $Q_2$ in $\F_{5}$, but does not fold into $Q_2$ in~$\F_{2}$ (see proof of Relation 9). In the proof of Relation 3 we showed that $P_2$ folds into~$Q_2$ in $\F_{3}$. As $ \F_{3}=\{+90^{\circ}\text{; interior faces; grid}\} \subset \F_{5}=\{\pm 90^{\circ}\text{; interior faces; grid}\},$ this holds true for $\F_{5}$.

{\bf Relation 11:} We show that $P_3$ folds into $Q_3$ in $\F_{5}$, but does not fold into $Q_3$ in $\F_{3}$ (see proof of Relation 8). In the proof of Relation 5 we showed that $P_3$ folds into $Q_3$ in~$\F_{2}$. As $ \F_{2}=\{\pm 90^{\circ}\text{; grid}\}\subset \F_{5}=\{\pm 90^{\circ}\text{; interior faces; grid}\},$ this holds true for $\F_{5}$.

{\bf Relation 12:} We show that $P_1$ folds into $Q_1$ in $\F_{7}$, but does not fold into $Q_1$ in $\F_{3}$ (see proof of Relation 7). In the proof of Relation 7 we showed that $P_1$ folds into $Q_1$ in $\F_{4}$. As $ \F_{4}=\{+90^{\circ}, 180^{\circ}\text{; grid}\}\subset \F_{7}=\{+90^{\circ}, 180^{\circ}\text{; interior faces; grid}\},$ this holds true for $\F_{7}$.

{\bf Relation 13:} We show that $P_2$ folds into $Q_2$ in $\F_{7}$, but does not fold into $Q_2$ in~$\F_{4}$ (see proof of Relation 6). In the proof of Relation 6 we showed that $P_2$ folds into~$Q_2$ in~$\F_{3}$. As $ \F_{3}=\{+90^{\circ}\text{; interior faces; grid}\}\subset \F_{7}=\{+90^{\circ}, 180^{\circ}\text{; interior faces; grid}\},$ this holds true for $\F_{7}$.

{\bf Relation 14:} We show that $P_3$ folds into $Q_3$ in $\F_{6}$, but does not fold into~$Q_3$ in~$\F_{4}$ (see proof of Relation 5). In the proof of Relation 5 we showed that $P_3$ folds into~$Q_3$ in~$\F_{2}$. As $ \F_{2}=\{\pm 90^{\circ}\text{; grid}\}\subset \F_{6}=\{\pm 90^{\circ}, 180^{\circ}\text{; grid}\},$ this holds true for~$\F_{6}$.

{\bf Relation 15:} We show that $P_1$ folds into $Q_1$ in $\F_{6}$, but does not fold into $Q_1$ in~$\F_{1}$ (see proof of Relation 4). In the proof of Relation 2 we showed that $P_1$ folds into~$Q_1$ in~$\F_{4}$. As $ \F_{4}=\{+90^{\circ}, 180^{\circ}\text{; grid}\}\subset \F_{6}=\{\pm 90^{\circ}, 180^{\circ}\text{; grid}\},$ this holds true for $\F_{6}$.

{\bf Relation 16:} We show that $P_2$ folds into $Q_2$ in $\F_{7}$, but does not fold into $Q_2$ in $\F_{6}$. No interior faces are allowed in $\F_{6}$; hence, Lemma~\ref{le:b-int} shows that $P_2$ does not fold into $Q_2$ in $\F_{6}$. In the proof of Relation 9 we showed that $P_2$ folds into $Q_2$ in $\F_{3}$. As $ \F_{3}=\{+90^{\circ}\text{; interior faces; grid}\}\subset \F_{7}=\{+90^{\circ}, 180^{\circ}\text{; interior faces; grid}\},$ this holds true for $\F_{7}$.

{\bf Relation 17:} We show that $P_3$ folds into $Q_3$ in $\F_{6}$, but does not fold into $Q_3$ in~$\F_{7}$. $|P_3|=46$, and $Q_3$ has $46$ square faces on its surface. Hence, all squares of $P_3$ are faces of $Q_3$. In the proof of Relation 8 we showed that $P_3$ does not fold into $Q_3$ in $\F_{3}=\{\text{grid: }+90^{\circ}\text{; interior faces}\}$. As $\F_{7}=\{\text{grid: }+90^{\circ}, 180^{\circ}\text{; interior faces}\},$ we need to use $180^{\circ}$ folds to fold $P_3$ into $Q_3$. But any such fold would cover a face of $Q_3$ twice, a contradiction to all squares of $P_3$ being faces of $Q_3$.
In the proof of Relation 8 we showed that $P_3$ folds into $Q_3$ in $\F_{2}$. As $ \F_{2}=\{\pm 90^{\circ}\text{; grid}\}\subset \F_{6}=\{\pm 90^{\circ}, 180^{\circ}\text{; grid}\},$ this holds true for $\F_{6}$. 

{\bf Relation 18:} We show that $P_3$ folds into $Q_3$ in $\F_{5}$, but does not fold into $Q_3$ in $\F_{7}$ (see proof of Relation 17).  In the proof of Relation 8 we showed that $P_3$ folds into $Q_3$ in $\F_{2}$. As $ \F_{2}=\{\pm 90^{\circ}\text{; grid}\}\subset \F_{5}=\{\pm 90^{\circ}\text{; interior faces; grid}\},$ this holds true for $\F_{5}$.

{\bf Relation 19:} We show that $P_1$ folds into $Q_1$ in $\F_{7}$ (see proof of Relation 12), but does not fold into $Q_1$ in $\F_{5}$. In the proof of Relation 8 we showed that $P_1$ does not fold into $Q_1$ in $ \F_{2}=\{\pm 90^{\circ}\text{; grid}\}$. As $\F_{5}=\{\pm 90^{\circ}\text{; interior faces; grid}\}$ we need to use interior faces. As $Q_1$ is a unit cube, we need two folds of $+60^{\circ}$, an infeasible grid fold in $\F_{5}$. 

{\bf Relation 20:} $P_1$ folds into $Q_1$ in $\F_{6}$ (see proof of Relation 15), but does not fold into $Q_1$ in $\F_{5}$ (see proof of Relation 19).

{\bf Relation 21:} $P_2$ folds into $Q_2$ in $\F_{5}$ (see proof of Relation 10), but does not fold into $Q_2$ in $\F_{6}$ (see proof of Relation 16). 

{\bf Relation 22:} We show that $P_2$ folds into $Q_2$ in $\F_{8}$, but does not fold into $Q_2$ in $\F_{6}$ (see proof of Relation 16). In the proof of Relation 3 we showed that $P_2$ folds into $Q_2$ in $\F_{3}$. As $ \F_{3}=\{+90^{\circ}\text{; interior faces; grid}\}\subset \F_{8}=\{\pm 90^{\circ}, 180^{\circ}\text{; interior faces; grid}\},$ this holds true for $\F_{8}$.

{\bf Relation 23:} We show that $P_1$ folds into $Q_1$ in $\F_{8}$, but does not fold into $Q_1$ in $\F_{5}$ (see proof of Relation 19). In the proof of Relation 2 we showed that $P_1$ does not fold into $Q_1$ in $ \F_{4}$. As $ \F_{4}=\{+90^{\circ}, 180^{\circ}\text{; grid}\}\subset \F_{8}=\{\pm 90^{\circ}, 180^{\circ}\text{; interior faces; grid}\},$ this holds true for $\F_{8}$.

{\bf Relation 24:} We show that $P_3$ folds into $Q_3$ in $\F_{8}$, but does not fold into $Q_3$ in $\F_{7}$ (see proof of Relation 17). In the proof of Relation 1 we showed that $P_3$ folds into $Q_3$ in $\F_{2}$. As $ \F_{2}=\{\pm 90^{\circ}\text{; grid}\}\subset \F_{8}=\{\pm 90^{\circ}, 180^{\circ}\text{; interior faces; grid}\},$ this holds true for $\F_{8}$.

\end{proof}

All but one of the polyominoes used in the proof of Theorem~\ref{le:hier} (Figure~\ref{fig:hierex}(a)-(d)) are tree shapes, that is, polyominoes whose dual graph of unit squares is a tree. We conjecture the following:

\begin{conj}\label{co:hier-tree}
The hierarchy in Theorem~\ref{le:hier} also holds for tree shapes. 
\end{conj}

\begin{figure}[t!]
\hspace*{.23\textwidth}
    \comic{.15\textwidth}{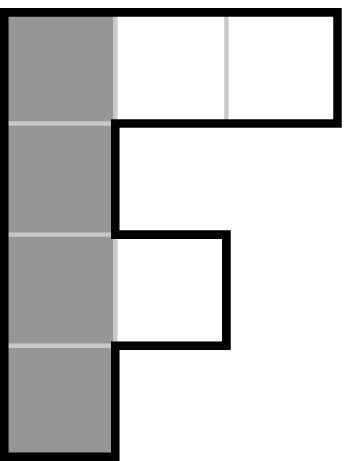}{(a)}
\hfill
\comic{.15\textwidth}{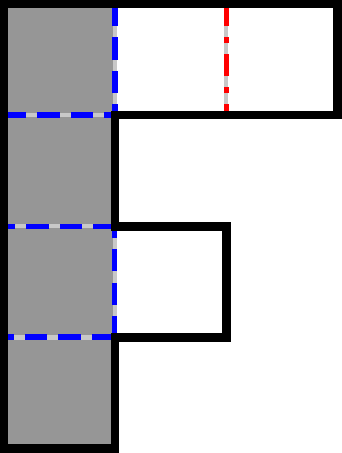}{(b)}\hspace*{.23\textwidth}
        \caption{\label{fig:mv} (a) Shape $P$ that needs both mountain and valley folds to cover a unit cube $C$. (b) The mountain and valley folds are shown in red and blue, respectively.}
\end{figure} 

\begin{proposition}\label{le:mv}
There exist tree shapes $P$ that need both mountain and valley folds to cover a unit cube $C$.
\end{proposition}


\begin{proof}
The shape $P$ shown in Figure~\ref{fig:mv}(a) does not fold into a unit cube $C$ with only valley folds: the four unit squares in the left column (gray) fold to a ring of size four, and the two flaps can only cover one of the two remaining cube faces. But with both mountain and valley folds, $P$ can be folded into $C$, by using a $180^{\circ}$ fold between the first column and the longer flap, see Figure~\ref{fig:mv}(b).
\end{proof}

We used example (d) in Theorem~\ref{le:hier} to show that some polyominoes require the use of interior faces in order to fold to a given polycube. The following theorem shows that $Q_2$, with size $2$, is a minimal example of a polycube such that folding it with a tree shape may require interior faces.


\begin{theorem}\label{le:interior}
Let $\F = \{$any$^{\circ}$; interior faces; grid$\}$.
Any tree shape $P$ that folds into a unit cube with interior faces also folds into a unit cube without interior faces.
\end{theorem}

\begin{proof}
We use induction on the size of $P$. Therefore, consider a $P$ that folds
into a unit cube with interior faces  such that any smaller tree that
folds into a unit cube with interior faces also folds into a unit cube
without interior faces. If $P$ folds into a unit cube with interior
faces in multiple ways, fix one of them with as few interior faces as
possible; if the number of interior faces is 0, then the induction is
complete. In the rest of this proof, ``the folding of $P$'' refers to
that folding.

The proof progresses as follows: first, we argue that we can fold away extraneous branches of the folded tree shape $P$ interior to the cube, leaving only topologically necessary interior faces. Then, we show that the interior faces must be connected to form long ``interior paths''. Finally, by restricting the ways these interior paths can turn and connect to the cube faces, we show by case analysis that the interior faces can be refolded to lie on the cube faces.

{\bf Claim 1: } No interior face of the folding of $P$ is a leaf of the dual of $P$.

If the folding of $P$ has some interior face $T$ that is a leaf of the
tree, then the restriction of the folding of $P$ to $P \setminus T$ also
folds into a cube (possibly with interior faces), so it folds into a
cube without interior faces. But if we fold $T$ onto its parent in the
tree and identify the two faces, we get exactly the restriction of the
folding of~$P$ to~$P \setminus T$, so we can use that folding to fold
all of $P$ into a cube without interior faces.

Since no interior face is a leaf, every interior face is part of some
path joining two cube faces in the folding of $P$. Call such a path
(including the two cube faces at its ends) an ``interior path''. When
convenient, we may refer to the unfolded state of such a path by the
sequence of directions from each square to the next, using N(orth) for
positive y, S(outh) for negative y, W(est) for negative x, and E(ast)
for positive x; for instance, the path of the three squares at the
grid positions (0,0), (0,1), and (1,1) could be written NE or WS.

{\bf Claim 2: } Every interior path has length at least 4.

Any single square which joins to two cube edges is one of the cube
faces. So there are at least two interior faces in any such path, plus
one face at each end, for a total of at least 4.

{\bf Claim 3: } Some interior path has at least one turn in the unfolded state.

Suppose to the contrary that every interior path has no turns in the unfolded state; that is, that all such
paths are $1 \times k$ rectangles in the unfolded state. We will eliminate
such paths one by one. After eliminating each path, we will have a
folding of $P$ into a cube with possibly-intersecting interior faces, but
after all interior paths are eliminated, there can be no more interior
faces that could intersect, and we will have a proper folding of $P$
(defined by fold angles at the edges) into a
cube without interior faces. This folding can always be done without
interference between cube face squares. The interior path splits the
polyomino into two pieces, each of which can be folded independently of
each other just as they were in the original folding. So we are left
with two partial cubes, connected by the interior path. If we invert
every mountain and valley fold in the folding of one of them, we get a
folding with inverted orientation (that is, with the innermost layer
outermost and vice versa). Then we can nest the folding of one part
inside the folding of the other part.

If we have an extended interior path that is a rectangle consisting of squares
$T_0$, $T_1$, $T_2$, $T_3$, \ldots, $T_k$ (where $T_0$ and $T_k$ are
cube faces), note that all the edges between squares on
the path are parallel, and the edge between $T_0$ and $T_1$ is a cube
edge, so all of them are parallel to the same cube edge. Without loss of
generality, let that cube edge be vertical. We will leave every square in the component of $P \setminus
T_1$ containing~$T_0$ in place, unfold the path so that it goes around
the faces of the cube other than the top and bottom (and, since it has
at least 4 faces by Claim 2, it covers all four such faces), and rotate
the rest of the unfolding of $P$ around the vertical axis of the cube so
that it still connects to $T_k$. The top and bottom faces are still
covered because they were covered in the folding of $P$, and every
existing cube face has either been left in place or rotated around the
vertical axis of the cube, both of which take the top face to the top
face and the bottom face to the bottom face.

After eliminating all such paths that way, we get a folding of $P$ into
a cube with no interior faces, as desired.

{\bf Claim 4:} In every interior path, there is at least one interior
square between each cube face to which the path connects and the first
square at which the path turns.

Suppose there is an interior path that turns at its first interior
square $T_1$. Let the orientation of the cube edge to which $T_1$
connects be $x$, let the orientation of the edge on which $T_1$ connects
to the next square $T_2$ of the path be $y$, and let the orientation of
the other edge of $T_2$ be $z$. $T_2$ has one vertex at a vertex $v$ of the
cube, since it shares an edge of $T_1$ adjacent (as an edge of $T_1$) to
the edge on which $T_1$ connects to a cube face. Without loss of
generality let the cube be axis-aligned and let $v$ be the vertex with
the least value in every coordinate. Then any edge inside the cube with
an endpoint at $v$ goes in a nonnegative direction in every
coordinate. In particular, $y$ and $z$ are two such directions, and
they are perpendicular, so their dot product is 0, so in every coordinate
at least one of them is 0. Since there are only three coordinates, one
of them is nonzero in at most one coordinate, that is, one of $y$ and
$z$ is the orientation of a cube edge. If $y$ is the orientation of a
cube edge, then $T_1$ has sides of orientations $x$ and $y$ and is
therefore a cube face, contradicting the definition of $T_1$. If $z$ is
the orientation of a cube edge and isn't $x$, then $y$ is perpendicular
to two cube edge orientations $x$ and $z$ and is therefore a cube edge
orientation. Finally, if $z=x$, then $T_1$ and $T_2$ are in the same
place in the folding, so if in the unfolded state we first fold $T_1$
onto $T_2$, we get a tree which folds into a cube with interior faces
and which therefore by induction folds into a cube without interior faces.

{\bf Claim 5:} If a square of an interior path is parallel to one of the
cube faces, then there are at least two interior squares between it and
any cube face to which the path connects.

Let $T$ be an interior square parallel to a cube face, and suppose that
exactly one interior face connects an edge $e$ of it to a cube edge
$e'$. Then $e'$ would either be a vertical translation of $e$ by a
distance strictly between 0 and 1 or a horizontal translation of such by
exactly 1, and in both cases the distance from $e$ to $e'$ is not
exactly 1. 

{\bf Claim 6:} No interior path with exactly one turn connects to
cube faces at edges of the same orientation.

Suppose to the contrary there is an interior path with exactly one turn that connects to
edges of the same orientation. Let $T$ be the square at which that path
turns. $T$ divides the path into two $1 \times k$ rectangles whose
intersection is $T$. In each such rectangle, the two edges of length 1
are parallel, since in each square of the rectangle, the two opposite
edges are parallel. Since $T$ connects to edges of the same orientation,
all four edges of length 1 in those rectangles are parallel; in
particular, two consecutive edges of $T$ are parallel. But they are
consecutive edges of the square, so they are perpendicular, a contradiction.

{\bf Claim 7:} No interior path with exactly one turn connects to cube
faces at edges of different orientations. 

Suppose to the contrary that such a path exists.  Let $T$
be the square at which the path turns. Consider the orientations of the
internal (folded) edges of that path: the first one is parallel to one
cube edge, and every folded edge before the turn is parallel to it, and
the last one is parallel to a differently oriented cube edge, and every
folded edge after the turn is parallel to it, so the two folded edges at
the turn are each parallel to one of the cube edges, so $T$ is parallel
to a cube face, without loss of generality the top one. Since $T$ is
contained in the cube, it is a vertical translation of the top face.

There are at least two interior faces between $T$ and each end of the path $R$ by
Claim~5. Then the interior path and the two cube faces to which it
connects form an $L$ shape with side length 4 squares. In this case we could
fold away every other square of $P$ (starting with the leaves) and
ignore them; an $L$ of side length 4 already folds into a cube.

{\bf Claim 8:} No interior path with exactly two turns connects to cube
faces at edges of different orientations.

Suppose to the contrary that there is an interior path with exactly two turns that connects to
edges of different orientations $x$ and $z$, and let the orientations of
the folded edges in the path be $x$, $y$, and $z$. Then $y$ is
perpendicular to both $x$ and $z$ (since the first turn square $T_i$ has
edges oriented $x$ and $y$ and the other $T_j$ has edges oriented $y$ and
$z$), so $y$ is also an orientation of a cube edge. Then $T_i$ is
interior to the cube and parallel to the $xy$ cube faces and $T_j$ is
interior to the cube and parallel to the $yz$ cube faces, so they
intersect, contradiction.

{\bf Claim 9:} No interior path with exactly three turns connects to
cube faces at edges of the same orientation.

Suppose to the contrary that such a path exists. Let the orientations of the folded edges be $x$, $y$, $z$,
and $x$ again. Then $x$ and $y$ are perpendicular, $y$ and $z$ are
perpendicular, and $z$ and $x$ are perpendicular, so the square with
edges oriented $y$ and $z$ is perpendicular to one cube edge (and hence
parallel to a cube face) and contained inside the cube, so it is a
translation of a cube face, that is, $y$ and $z$ are the other two
orientations of cube edges. Then the square at which the path first
turns has edges $x$ and $y$, so it is a translation of the $xy$ cube
face, and the square at which the path turns second has edges $y$ and
$z$, so it is a translation of the $yz$ cube face. Both of those squares
are inside the cube, so they intersect, a contradiction.

{\bf Claim 10:} No interior path with exactly three turns connects to
cube faces of different orientations.

Suppose to the contrary that such a path exists. Let the orientations of the folded edges be $x$, $y$, $z$,
and $w$. Then $x$ and $y$ are perpendicular, $y$ and $z$ are
perpendicular, $z$ and $w$ are perpendicular, and $w$ and $x$ are
perpendicular (since they are two distinct cube edge orientations), so
each of $x$ and $z$ is perpendicular to each of $y$ and $w$. But we're
working in 3 dimensions, so the span of $x$ and $z$ and the span of $y$
and $w$ cannot both be two-dimensional, that is, either $x=z$ or $y=w$.
Without loss of generality $y=w$, and the orientations of the folded
edges of the path are $x$, $y$, $z$, and $y$. By Claim 5, there are at
least 4 squares in a line up to the square spanned by $x$ and $y$, and
by Claim 4, there are at least three squares in a line starting from the
last square spanned by $z$ and $y$. That is, the unfolded state
contains, after folding away any extra squares in those lines, either
NNNENWW or NNNENEE or NNNESEE or NNNESWW, up to rotation and reflection.
One fold takes the middle two to NNNEEE, which folds into a cube; one
fold takes the other two to a line of four with one square on either
side, which folds into a cube.

{\bf Claim 11:} No interior path has an even number of turns that is greater than or equal to~$4$.

Suppose to the contrary that such a path exists. By Claim 4, there are at least three squares in a line up
to and including the first square at which there is a turn and three
squares in a line from the last square at which there is a turn. Also,
there are at least two squares on the path between those two lines of
length 3, since there are at least two more squares at which the path
turns. Let the lines of length~3 be oriented east-west in the unfolded
state. If they are separated by at least two spaces north-south, then we
can fold over any north-south folds (between squares adjacent east or
west) to get a north-south line of length 4 with two lines extending by
two to the east or west; by folding one of those over if necessary, we
get a line of length 4 with a square on either side, which folds into a
cube. If there are two lines of length 3 that cover at least 5 spaces
east-west, then by folding east-west folds in the middle we get EENEE,
which folds into a cube. Otherwise, we have two lines of length 3
separated by at most one space north-south and covering at most 4 spaces
east-west. By folding any east-west folds north of the north line or
south of the south line, we get, up to rotation or reflection, either
EESWSEE or EESWSWW or EESWWSEE(E). The first two fold into WSWSW, which
folds into a cube; the last folds into WSWWS, which folds into a cube.

{\bf Claim 12:} No interior path has an odd number of turns that is greater than or equal to~$5$.

Suppose to the contrary that such a path exists. First, if there are any squares in the interior of the path
at which turns don't occur, fold them over. Without loss of generality,
let the interior path start at (0,0) and go EES to (1,0), (2,0),
(2,-1). If the path ever extends two spaces farther E in the unfolded
state, folding over all north-south folds past that EES that gives at
least EESEE, which folds into a cube. If the path ever again extends
farther W than it started in the unfolded state, then folding over all
north-south folds past that EES gives at least EESWWW, which folds into
EESEE, which folds into a cube. So, if the unfolded state ever extends
farther N than 1, it does so via the spaces (3,-1) and (3,0). Also, if
the unfolded state extends as far N as 2, then all three spaces (3,-1),
(3,0), (3,1) are used and in a straight line, so by folding all
east-west folds after (3,-1) and all north-south folds before it, we get
EEENNN, which folds into a cube. If the unfolded state extends farther N
than 0, then, since the unfolded state ends with a N-S line of three
squares (since there are an odd number of turns), can't go back below
N=0 once it is gone above it, and doesn't have room for a N-S line of
length three ending anywhere but (3,1), the unfolded state ends at
(3,1). Since the path has length at least 9, (2,-1) doesn't connect
directly to (3,-1), so all four spaces (3,-2), (3,-1), (3,0), and (3,1)
are used and in a straight line, and by folding all east-west folds
after (3,-1) and all north-south folds before it, we get EEENNN, which
folds into a cube. If the path extends both as far E as 3 and as far S
as -3, then by folding anyway anything W of 2 by north-south folds, we
get (E)ESSSE, which folds into a cube. If the path extends as far E as
3, then it fits in a $3 \times 4$ bounding box and has at least 5 turns,
and the only two remaining paths are EESWWSEEENN and EESWSEENN. The
first can be folded into the second, which can be folded into EESSENN,
which can be folded into EESSES, which can be folded into ESSES, which
folds into a cube. In the only case left, the unfolded state is only
three squares wide east-west. If any of the other cases apply starting
from the other end of the path, we also fold into a cube; otherwise, the
unfolded state fits in a $3 \times 3$ square, but it is at least 9 long
and starts and ends with perpendicular paths of length 3, which don't
fit in a $3 \times 3$ square, contradiction. Hence any such path can be
folded into a cube.

{\bf Claim 13:} No interior path with exactly two turns connects to cube
faces at edges of the same orientation.

If there is such a path $R$, say it connects to cube faces at edges of
the same orientation $x$. Let the orientations of the folded edges
be x, y, and x, let $T_i$ and $T_j$ be the squares at which the turns
occur (whose edges are oriented $x$ and $y$), and let $T_0$ and $T_k$
be the cube faces at the ends of the paths. Then $2 \le i < j \le k-2$
by Claim 4. Without loss of generality, let the first turn in
the unfolded state be a right turn.

If the second turn is a left turn,
then fold away everything but the path, and fold away the portion of the
path between $T_i$ and $T_j$ to get a path of length 6 that goes NNENN,
which folds into a cube.

If both turns are right turns and either $i > 3$ or $j \le k-3$, then we
can fold away the portion of the path between $T_i$ and $T_j$ and fold
the other edge adjacent to either $T_i$ or $T_j$ to get a path of length
6 that goes NNENN, which folds into a cube.

If both turns are right turns and $j - i \ge 3$, then by folding
everything else away and folding the edge between $T_i$ and $T_i-1$ we
get a path that goes SEEES, which folds into a cube.

If both turns are right turns, $j - i = 1$, and either end of the path
(without loss of generality the $T_k$ end of the path) connects to more than one square
(besides $T_0$ or $T_k$), then the path contains either N, N, E, S*, S,
S or N, N*, E, S, S, E, and making
180$^{\circ}$ folds over the starred edges gives a shape that folds into
a cube.

If both turns are right turns, $j-i = 1$, and neither end of the path
connects to more than one square in the folded state, then there are
more squares in the unfolded state than the six squares of the path $T_0 = (0,0)$, $T_1 =
(0,1)$, $T_2 = (0,2)$, $T_3 = (1,2)$, $T_4 = (1,1)$, $T_5 = (1,0)$,
because those only cover 2 cube faces. By the previous paragraph, the
rest of the unfolding doesn't connect at $T_0$ or $T_5$, and by Claim 4,
nothing connects at $T_1$ or $T_4$. Suppose some path $R$ leading to
another cube face connects to $(0,2)$. If the next square on $R$ is
$(0,3)$, then the path $T_5$, $T_4$, $T_3$, $T_2$, $R$ is an interior
path with at least two turns and its first two turns in different
directions, but we've already eliminated all such cases. Otherwise, the
next square of $R$ is $(-1,2)$. If the following square is $(-1,3)$,
then the path $T_5$, $T_4$, $T_3$, $T_2$, $R$ is an interior path with
at least two turns and its first two turns in different directions, but
we've already eliminated all such cases. If the following square is
$(-2,-2)$, then the path $T_5$, $T_4$, $T_3$, $T_2$, $R$ is an interior
path with at least one turn and at least three spaces in a line after
its first turn, but we've already eliminated all such cases. So the
following square is $(-1,1)$. That accounts for both turns on the path,
so the next square is $(-1,0)$. If that is not a cube face, then the path
$T_5$, $T_4$, $T_3$, $T_2$, $R$ has two right turns and either $i > 3$
or $j \le k-3$, an already-eliminated case. Similarly, if there is more
than one cube face there, then the path $T_0$, $T_1$, $T_2$, $R$ is a
path with two right turns,  $j - i = 1$, and an end of the path
connecting to more than one square, an excluded case. So we've only accounted
for three cube faces and need at least three more, but there is nowhere
left that a path to a cube face can attach, contradiction.

If both turns are right turns, $j - i = 2$, and both ends of the path
connect to more than one square (besides $T_0$ or $T_k$): in the
unfolded state, $T_0$ and $T_k$ have exactly one square between them,
and that square isn't the extra square for both ends of the path because
the unfolded state is a tree, so some end of the path (without loss of generality the $T_k$
end of the path) connects to a different other square, that is, the path
contains either N, N, E, E*, S*, S, S or N, N*, E*, E, S, S, E, and making 180$^{\circ}$ folds over
the starred edges gives a shape that folds into a cube.

If both turns are right turns, $j-i = 2$, and one end of the path
(without loss of generality the $T_k$ end) connects to no cube faces
except its end ($T_k$) in the folded state, then nothing else connects
to the interior of the path, because we've already eliminated all cases
for interior paths except this one and straight interior paths, but
there is nowhere we can connect an interior path to this one without
creating at least one non-straight interior path. We will refold the
squares of the interior path to leave $T_0$ fixed and nevertheless cover
the cube face that $T_k$ covered; this doesn't affect any other cube
faces. In fact, starting from $T_0$, a path N, N, E, E, S, S along the cube covers four of its faces, and folding the first
E and the first S 180$^{\circ}$ to get a path N, N, W, N covers the remaining two. So we can cover the whole cube without that
interior path, as desired.
\end{proof}


\section{Polyominoes that fold into a cube}\label{sec:char}

In this section we characterize all polyominoes that can be folded into a unit cube and all tree-shaped polyominoes that are a subset of a 2$\times n$ or 3$\times n$ strip and can be folded into a unit cube using arbitrary grid folds.
In Section~\ref{sec:const} we present exhaustive results for all polyominoes of constant size obtained by computer search. 

\begin{theorem}\label{le:w45}
Consider a polyomino $P$ of size $|S|=n$ and a unit cube $C$ under a folding model $\F = \{$any$^{\circ}$; grid; diagonal; half-grid$\}$, 
such that each face of $C$ has to be covered by a full unit square of $P$.
Then 
there is a polyomino of size $n=9$ that \emph{cannot} be folded into $C$, 
while all polyominoes with $n\geq 10$ \emph{can} be folded. 
\end{theorem}

\begin{figure}[t]
\centering
\subfloat[]{\label{fig:ub10a-a}
\includegraphics[page=1]{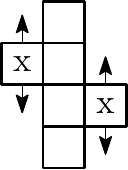}}
\hfil
\subfloat[]{\label{fig:ub10a-b}
\includegraphics[page=2]{ub10-new}}
\hfil
\subfloat[]{\label{fig:ub10a-c}
\includegraphics[page=3]{ub10-new}}
\hfil
\subfloat[connectivity transformation]{\label{fig:ub10a-d}
\includegraphics[page=4]{ub10-new}}
\\
\subfloat[number transformation]{\label{fig:ub10a-e}
\quad\includegraphics[page=5]{ub10-new}\quad}
\hfil
\subfloat[corner transformation]{\label{fig:ub10a-f}
\qquad\includegraphics[page=6]{ub10-new}\qquad}
\hfil
\subfloat[width/height transformation]{\label{fig:ub10a-g}
\includegraphics[page=7]{ub10-new}}
\caption{Proof details for Theorem~\ref{le:w45}. The mountain and valley folds are shown in red and blue, respectively. Unit squares marked with ``x'' can be located at an arbitrary height and light gray squares fold away. For the two adjacent corner transformations in the left of~\protect\subref{fig:ub10a-f}, fold away the upper shaded corner before folding away the lower one.}
\label{fig:ub10a}
\end{figure}


\begin{proof}
For the lower bound see Figure~\ref{fig:lb9}. The polyomino in Figure~\ref{fig:lb9-a} cannot be folded into a unit cube under our model. Note that if we allow half-grid folds without requiring faces of $C$ to be covered by full unit squares of $P$, 
we \emph{can} turn this shape $P$ into a unit cube; see Figure~\ref{fig:lb9-b}.

For the upper bound $n\geq 10$, we start
by identifying several \emph{target polyominoes}, examples of which are shown in
Figure~\ref{fig:ub10a-a}--\subref{fig:ub10a-c}.
Each target polyomino can be folded into a cube using only grid folds. We denote them by listing the numbers of squares in each column, using + or -- to signify the cases where the columns may be attached in arbitrary arrangements or a specific arrangement is required, respectively. The target polyominoes are:
\begin{enumerate}
\item[(a)] A \emph{1+4+1 polyomino}, composed of one contiguous column of four unit squares, with one more unit square on either side, at an arbitrary height.
\item[(b)] A \emph{2--2--2 polyomino}, composed of three (vertical) pairs attached in the specific manner shown.
\item[(c)] A \emph{2--3+1 polyomino}, composed of a (vertical) pair and triple attached in the specific manner shown, with one more unit square at an arbitrary height.
\end{enumerate}

\begin{figure}[t]
\centering
\setcounter{subfigure}{7}
\subfloat[]{\label{fig:ub10a-h}
\includegraphics[page=8]{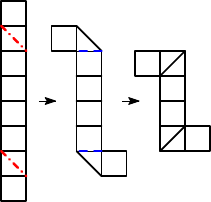}}
\hfil
\subfloat[]{\label{fig:ub10a-i}
\includegraphics[page=9]{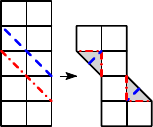}}
\\
\subfloat[]{\label{fig:ub10a-j}
\includegraphics[page=10]{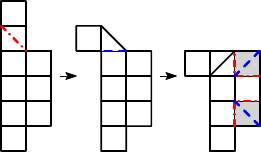}}
\hfil
\subfloat[]{\label{fig:ub10a-k}
\includegraphics[page=11]{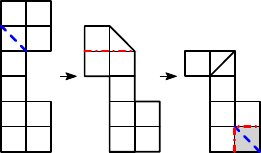}}
\caption{Proof details for Theorem~\ref{le:w45}. The mountain and valley folds are shown in red and blue, respectively. Regions shaded in light gray are folded away.}
\label{fig:ub10a2}
\end{figure}

We next describe some transformations that can be used to convert any polyomino of sufficient size to one of the target polyominoes. See Figure~\ref{fig:ub10a-d}--\subref{fig:ub10a-g} for examples of the following.
If $n_i$ is the maximum number of unit squares in any column (say,
in $i$), we can apply a \emph{connectivity transformation} \subref{fig:ub10a-d} by using 
(horizontal) half-grid folds to convert $P$ into a
polyomino $P'$ in which these $n_i$ unit squares form a contiguous
set, while leaving at least one unit square in each previously occupied column.
A \emph{number transformation} \subref{fig:ub10a-e} lets us fold away extra unit squares for $n>10$. 
\emph{Corner transformations} \subref{fig:ub10a-f} fold away unnecessary unit squares (or
half-squares) in target shapes by using diagonal folds when turning them into a
unit cube. Finally, \emph{width} and \emph{height transformations} \subref{fig:ub10a-g} 
fold over whole columns or rows of $P$ onto each other, producing a connected polyomino
with a smaller total number of columns or rows.

We consider a bounding box of $P$ of size $X\times Y$, with
$X\leq Y$. The unit squares are arranged in \emph{columns} and
\emph{rows}, indexed $1,\ldots, X$ and $1,\ldots, Y$, 
with $n_i$ unit squares~$P$ in column~$i$, and $m_j$ 
unit squares in row $j$.

Now consider a case distinction over $X$; see Figure~\ref{fig:ub10a-h}--Figure~\ref{fig:ub10b-t}.
For $X=1$, the claim is obvious, as we can reduce $P$ to a 1+4+1 target; see, for example, Figure~\ref{fig:ub10a-h}. 
For $X=2$, note that $Y\geq5$ and assume that $n_1\geq n_2$. If $n_1\geq8$, 
a width transformation yields the case $X=1$, so 
assume $n_1\leq 7$, and therefore $n_2\geq3$. By a number transformation, we
can assume $n_2\leq5$. If $P$ is a $2\times 5$ polyomino, we can make use of a 
1+4+1 polyomino with corner transformations; see Figure~\ref{fig:ub10a-i}. 
If~$n_1>n_2$, we have $n_1\geq 6$.
Because $P$ is connected, any two units squares in
column~1 must be connected via column~2, requiring at least three unit squares;
because~$n_2\leq 4$, we conclude that column~1 contains at most two connected
components of unit squares. Thus, at most one connectivity transformation makes column~1 connected, with
$n'_2\in \{n_2-1, n_2\}$ unit squares in column~2. Possibly using height
transformation, we get a connected polyomino $P''$ with vertical size six, six
unit squares in column~1, and $n''_2\in\{1,2,3,4\}$ unit squares in column~2. For
$n''_2\in\{1,2, 3\}$, there is a transformation to target shape
1+4+1 even if the squares in column~2 form a single connected component; for example see Figure~\ref{fig:ub10a-j}. For $n''_2=4$, a similar transformation 
exists; see, for example, Figure~\ref{fig:ub10a-k}. This leaves $n_1=5$, and thus $n_2=5$,
without $P$ being a $2\times 5$ polyomino. 
This  maps to a 1+4+1 target shape,
as shown for example in Figure~\ref{fig:ub10b-l} by folding a unit square from column
2 that extends beyond the vertical range of the unit squares in column~1 over to 
column~0.
%
\begin{figure}
\centering
\setcounter{subfigure}{11}
\subfloat[]{\label{fig:ub10b-l}
\includegraphics[page=12]{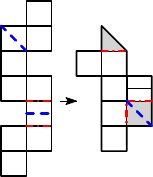}}
\hfil
\subfloat[]{\label{fig:ub10b-m}
\includegraphics[page=13]{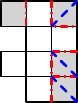}}
\hfil
\subfloat[]{\label{fig:ub10b-o}
\includegraphics[page=14]{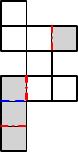}}
\\
\subfloat[]{\label{fig:ub10b-p}
\includegraphics[page=15]{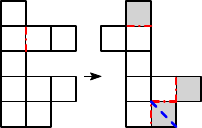}}
\hfil
\subfloat[]{\label{fig:ub10b-q}
\includegraphics[page=16]{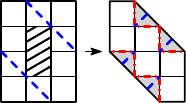}}
\hfil
\subfloat[]{\label{fig:ub10b-r}
\includegraphics[page=17]{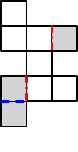}}
\\
\subfloat[]{\label{fig:ub10b-s}
\includegraphics[page=18]{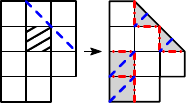}}
\hfil
\subfloat[]{\label{fig:ub10b-t}
\includegraphics[page=19]{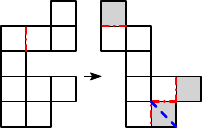}}
\caption{Proof details for Theorem~\ref{le:w45}. The mountain and valley folds are shown in red and blue, respectively. Regions shaded in light gray are folded away and the hatched region denote holes.}
\label{fig:ub10b}
\end{figure}

For $X=3$ and $n_2\geq4$, we can use 
connectivity, height and width transformations to obtain a new polyomino $P'$ that
has height four, a connected set of $n'_2=4$ unit squares in column~2 and
$1\leq n'_1 \leq 4$, as well as $1\leq n'_3 \leq 4$ unit squares in columns 1
and 3. Note that at least one square must remain in each of column~1 and column~3 throughout these transformations.
This easily converts to a 1+4+1 shape, possibly with corner transformations;
as in Figure~\ref{fig:ub10b-m} for example. Therefore, assume~$n_2\leq 3$ and (without loss of generality) $n_1\geq n_3$, implying
$n_1\geq 4$. 
If $n_1\geq 5$ and column~2 is connected, then $P$ contains a 2--3+1 target shape,
see Figure~\ref{fig:ub10b-o} for an example; if column~2 is disconnected, we can use a
vertical fold to flip one unit square from column~3 to column~0, obtaining a 1+4+1 target shape,
similar to Figure~\ref{fig:ub10b-p}. 
If $n_1 = 4$, $n_2 = 1$, and $n_3 = 5$ then a diagonal fold of the square in column~2, along with a corner transformation maps to a 1+4+1 target shape. As a consequence, we are left with $n_1=4$, $2\leq n_2 \leq 3$, $3\leq n_3 \leq
4$, where $n_2+n_3=6$.
If $n_2=2$, the unit squares
in columns 1 and 3 must be connected. For this it is straightforward to check
that we can convert $P$ into a 2--2--2 target polyomino; see Figure~\ref{fig:ub10b-q} for example. 
Therefore, consider $n_2=3$, $n_3=3$. This implies that column~1 contains at most two
connected sets of unit squares. If there are two, then the unit squares in
column~2 must be connected, implying that we can convert $P$ into a 2--3+1
target shape; as in Figure~\ref{fig:ub10b-r} for example. 
Thus, the four unit squares in column~1 must be connected. If the
three unit squares in column~2 are connected, we get a 2--3+1 target shape; see
for example Figure~\ref{fig:ub10b-r}. If
the unit squares in column~2 are disconnected, but connected by the three unit
squares in column~3, we convert this to a 2--3+1 target shape; see
Figure~\ref{fig:ub10b-s}. This leaves the
scenario in which there is a single unit square in column~1 whose removal
disconnects the shape; for this we can flip one unit square from column~3 to column~0 in order
to create a 1+4+1 target shape; see Figure~\ref{fig:ub10b-t} for an example.

For $X\geq 4$, we proceed along similar lines. If there is a row or column
that contains four unit squares, we can create a 1+4+1; otherwise,
a row or column with three unit squares allows generating a 2--3+1. If there is no such
row or column, we immediately get a 2--2--2. 
\end{proof}


In order to classify which tree shapes $P$ fold to a unit cube, we define two infinite families of polyominoes. First, the family given by a $1 \times n$ strip, with any number of added tabs consisting of squares adjacent to the same side of the strip and potentially including their left and/or right neighbors (see Figure~\ref{fig:2tree}). The second family consists of a $1 \times n$ strip with a single tab of height $2$, also potentially with their left and/or right neighbors (Figure~\ref{fig:3tree}).

\begin{theorem}\label{le:wo45}Given a tree shape $P$, 
a unit cube $C$ and $\F = \{$any$^{\circ}$; grid$\}$.
\begin{itemize}
\item[(a)] If $P$ is a subset of a $2\times n$ strip, then only the infinite family defined by Figure~\ref{fig:2tree} cannot fold into $C$.
\item[(b)] If $P$ is a subset of a $3\times n$ strip, then only the infinite family defined by Figure~\ref{fig:3tree} cannot fold into $C$.
\end{itemize}
\end{theorem}

\begin{figure}
\centering
\includegraphics[width=0.4\textwidth]{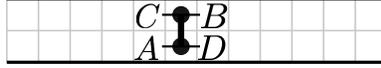}
\caption{\label{fig:2-vert} A vertical edge in a subset of a $2\times n$ strip with possible adjacent vertices for subtrees.}
\end{figure}
\begin{proof}
For the subset of a $2\times n$ strip consider one vertical edge, as shown in Figure~\ref{fig:2-vert}, and the possible subtrees attached at $\alpha$, $\beta$, $\gamma$ and $\delta$. One such vertical edge has to exist, otherwise the strip is a $1\times n$ strip and never folds to a cube.
We consider the length of subtrees attached at $\alpha$, $\beta$, $\gamma$ and $\delta$ when folded to the same row as this ``docking'' unit square to the vertical edge. With slight abuse of notation we refer to these lengths by $\alpha$, $\beta$, $\gamma$ and $\delta$ again.

\begin{figure}
\centering
\comic{.17\textwidth}{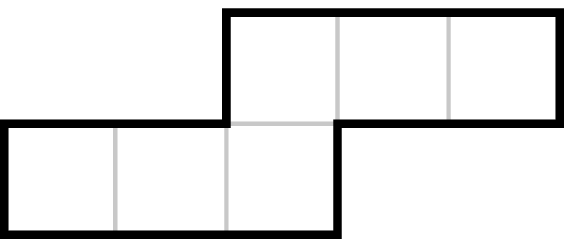}{(a)}\hfill
    \comic{.1\textwidth}{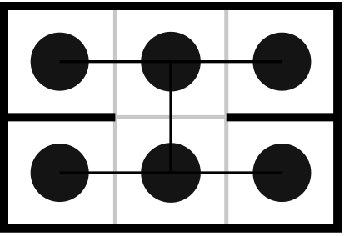}{(b)}\hfill
\comic{.13\textwidth}{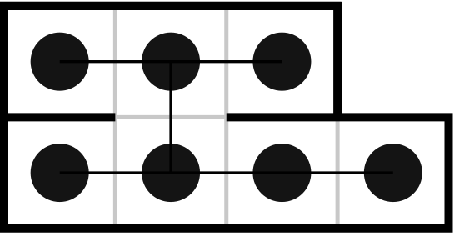}{(c)}\hfill
\comic{.13\textwidth}{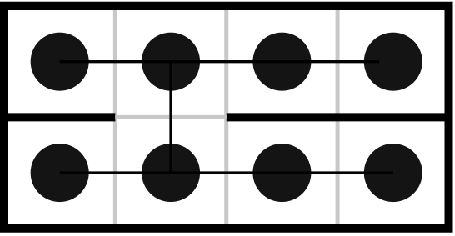}{(d)}\hfill
\comic{.17\textwidth}{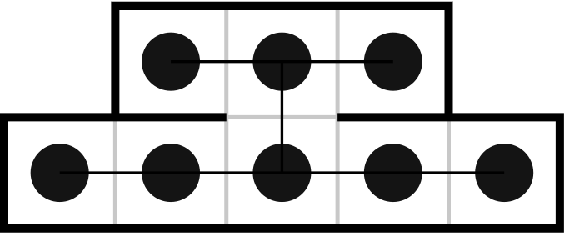}{(e)}
         \caption{\label{fig:abcd} In $\F = \{$any$^{\circ}$; grid$\}$ only the shapes in (a) and (d) folds into a unit cube $C$, the shapes in (b), (c) and (e) do not fold into $C$.}
\end{figure}

\def\AND{\;\mbox{\textsc{and}}\;}
\def\OR{\;\mbox{\textsc{or}}\;}

The first observation is that, 
if $(\alpha\geq2 \AND \beta\geq2) \OR (\gamma\geq2 \AND \delta\geq2)$, then $P$ folds to a cube. In this case we can fold all other squares away to end up with the shape in Figure~\ref{fig:abcd}(a).  
Not included in this categorization are the four shapes shown in Figure~\ref{fig:abcd}(b)--(e). Of those only (d) folds into a cube.

\begin{figure}[bt]
\centering
    \comicII{\columnwidth}{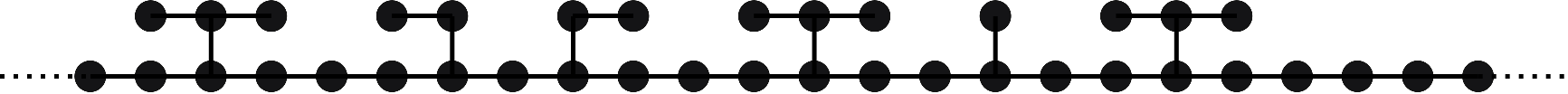}
         \caption{\label{fig:2tree} An example polyomino from an infinite family that cannot fold into a unit cube.
         }
\vspace*{-.1cm}
\end{figure}

\begin{figure}[b]
\centering
    \comicII{\columnwidth}{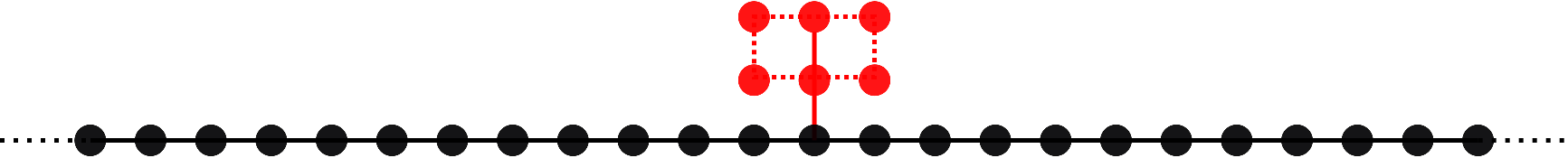}
         \caption{\label{fig:3tree} Infinite family that cannot fold into a unit cube, with at least a height 2 red subset (others are optional).
         }
\end{figure}

Thus, if we characterize more precisely which cases yield the shapes from Figure~\ref{fig:abcd}(a) and (d), we obtain a cube for:
\begin{align*}
\{(\alpha\geq 2 \OR \delta\geq3)  \AND (\beta\geq 2 \OR \gamma\geq3)\}&\OR\\
\{(\alpha\geq 3 \OR \delta\geq2)  \AND (\beta\geq 3 \OR \gamma\geq2)\}&\OR\\
\{(\alpha\geq 2 \AND \gamma\geq3)  \OR (\alpha \geq 3 \AND \gamma\geq2)\}&\OR\\
\{(\beta\geq 2 \AND \delta\geq3)  \OR (\beta\geq 3 \AND \delta\geq2)\}&\OR\\
\{(\alpha\geq 1 \AND \gamma\geq1 \AND \beta\geq 2 \AND \delta\geq2)\}&\OR\\
\{(\alpha\geq 2 \AND \gamma\geq2 \AND \beta\geq 1 \AND \delta\geq1)\}&\,.
\end{align*}

Figure~\ref{fig:2tree} shows an example of a $2 \times n$ polyomino that does not satisfy the above conditions and thus does not fold to a unit cube.

For the subset of a $3\times n$ strip,
we denote the rows of the height $3$ strip by row 1, row 2 and row 3. 
In case we have more than three unit squares in two different rows, we can fold over squares to obtain the shape from Figure~\ref{fig:abcd}(a), which folds into a unit cube. Otherwise, only one row spans the length of the strip.
If there are vertical edges adjacent to a $1\times n$ strip to two different sides (above and below), that is, the $1\times n$ strip is located in row 2, this folds to a cube. Consequently, if we have height three and cannot fold the shape into a unit cube, a long $1\times n$ strip cannot be located in the center row, row 2, of the shape. 
Without loss of generality, let the $1\times n$ strip be located in the lowest row, row 1. 
As we have height three, there is at least
a height two part which is a subset of the red part in Figure~\ref{fig:3tree} connected to the strip. That is, there are unit squares in row 2 and row 3, occupying a connected subset of the red part in  Figure~\ref{fig:3tree}. Let these unit squares be denoted by $P_{\mbox{red}}$.
If there exists another vertical edge of length at least one, $P_1$ that is only adjacent to the $1\times n$ strip (but not directly to $P_{\mbox{red}}$) we can fold over $P_{\mbox{red}}$ and obtain a case from above which can easily be folded to a unit cube. Consequently, only a single  vertical subset of length two can be attached, as shown in red in Figure~\ref{fig:3tree}.%
%
%
\end{proof}


\subsection{Enumeration of cube-foldable polyominoes}\label{sec:const}

In this section we present results on folding polyominoes of constant
size into a cube. We consider polyominoes of up to 14 unit squares.
Our
results have been obtained by exhaustive computer search. 
The second column of Table~\ref{tab:smallcubes} shows the number of different
\emph{free polynominoes} with a given number of squares. Here, ``free'' means that elements which can be transformed into
each other by translation, rotation and/or reflection are counted only
once.
For polyominoes whose dual (c.f. Section~\ref{sec:not}) is not a tree, that is, their connecting graph
contains cycles, we considered all possible dual trees. 
Thus, we first
generated all such dual trees for polyominoes of size up to 14. 
The third column 
of Table~\ref{tab:smallcubes} shows their number,
compared to the number of different free polyominoes.
While the number of different free polyominoes is currently
known for shapes of size up to 28 (\url{http://oeis.org/A000105}), the
number of different dual trees was known only for up to 10 elements
(\url{http://oeis.org/A056841}).

\begin{table}
\centering
{
\begin{tabular}{|r||*{6}{r|}}
\hline
$n$ & \multicolumn{1}{c|}{free}  & \multicolumn{1}{c|}{dual}  & \multicolumn{3}{c|}{foldable}     & not \\
\cline{4-6}
    & polyominoes &\multicolumn{1}{c|}{trees}   & with $\pm 90^{\circ}$  & and $\pm 180^{\circ}$ & and diagonal  &  \\
\hline
\hline
 2 &       1 &        1 & & & & \\
 3 &       2 &        2 & & & & \\
 4 &       5 &        5 & & & & \\
 5 &      12 &       15 & & & & \\
 6 &      35 &       54 &      11 &      0 &    0 & 43 \\
 7 &     108 &      212 &      90 &     24 &   39 & 59 \\
 8 &     369 &      908 &     571 &    175 &  126 & 36 \\
 9 &    1285 &     4011 &    3071 &    697 &  233 & 10 \\
10 &    4655 &    18260 &   15645 &   2230 &  385 &  0 \\
11 &   17073 &    84320 &   77029 &   6673 &  618 &  0 \\
12 &   63600 &   394462 &  374066 &  19337 & 1059 &  0 \\
13 &  238591 &  1860872 & 1803568 &  55477 & 1827 &  0 \\
14 &  901971 &  8843896 & 8682390 & 158208 & 3298 &  0 \\
\hline
\end{tabular}
}%
\caption{Different ways of folding small polyominoes into a cube.}
\label{tab:smallcubes}
\end{table}

We checked for each of the generated dual trees of a polyomino (of size six or more) whether it can be folded into a unit cube.
We did this in three different
steps. In the first step, we only allowed $\pm 90^{\circ}$ folds. Column~4 of
Table~\ref{tab:smallcubes} shows how many (dual trees of) polyominoes
can be folded to a cube with this restriction. It is interesting to observe that
while for $n=6$ only 11 polyominoes can be folded to a cube, for $n=14$ over $98\%$ of all shapes can be folded to a cube.

In the second step we tried $\pm 90^{\circ}$ and $\pm 180^{\circ}$ folds for
the remaining dual trees. Table~\ref{tab:smallcubes} gives in column~5
how many additional cube foldings can be obtained this way. For the dual trees for which this step found a folding, we never needed to use more than two $\pm 180^{\circ}$ folds. For $n=11$ and
$n=12$ there are each only one example which needs two $\pm 180^{\circ}$
folds, and we found no such 
examples for $n \geq 13$: for those values of $n$ all foldable examples could be folded with just one
$\pm 180^{\circ}$ fold.

In the last step we allowed also diagonal folds for the remaining dual trees. 
Column~6 of Table~\ref{tab:smallcubes} shows how many additional dual trees can be folded 
in this case,
and the last column 
gives the number of remaining
(non-foldable) dual trees. The most interesting result here is that
for $n \geq 10$ all polyominoes, regardless of which dual tree we select
for them, can be folded in this way. This is similar to Theorem~\ref{le:w45}, except that here we do not allow half-grid folds, but instead allow covering a cube face with triangles from diagonal folds.
Moreover, all such foldings need at
most one diagonal fold, with the exception of the $7\times 1$ strip,
the only example we found that needs two diagonal folds.

\begin{figure}
\centering
\includegraphics[width=\textwidth]{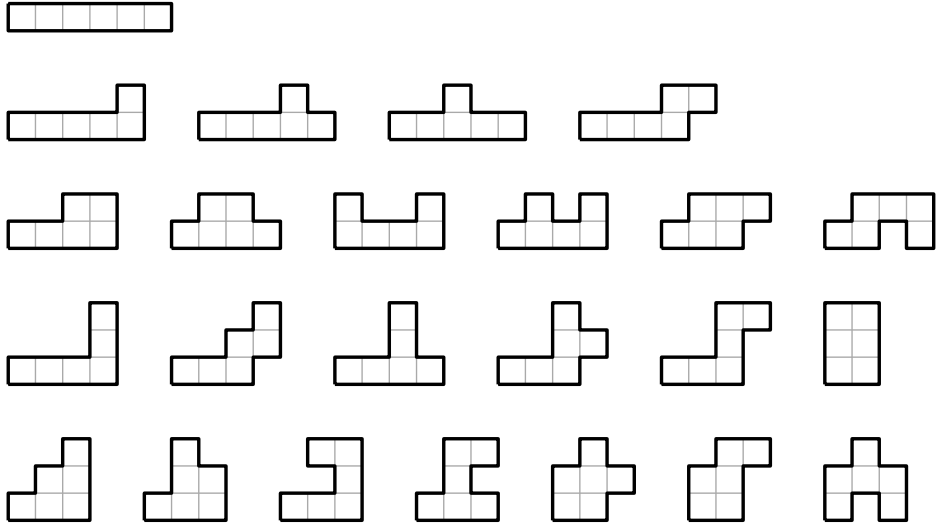}
\caption{Polyominoes of size 6 that cannot be folded to a cube (for any dual tree).}
\label{fig:polyominoes06nocube}
\end{figure}

\begin{figure}
\centering
\includegraphics[width=\textwidth]{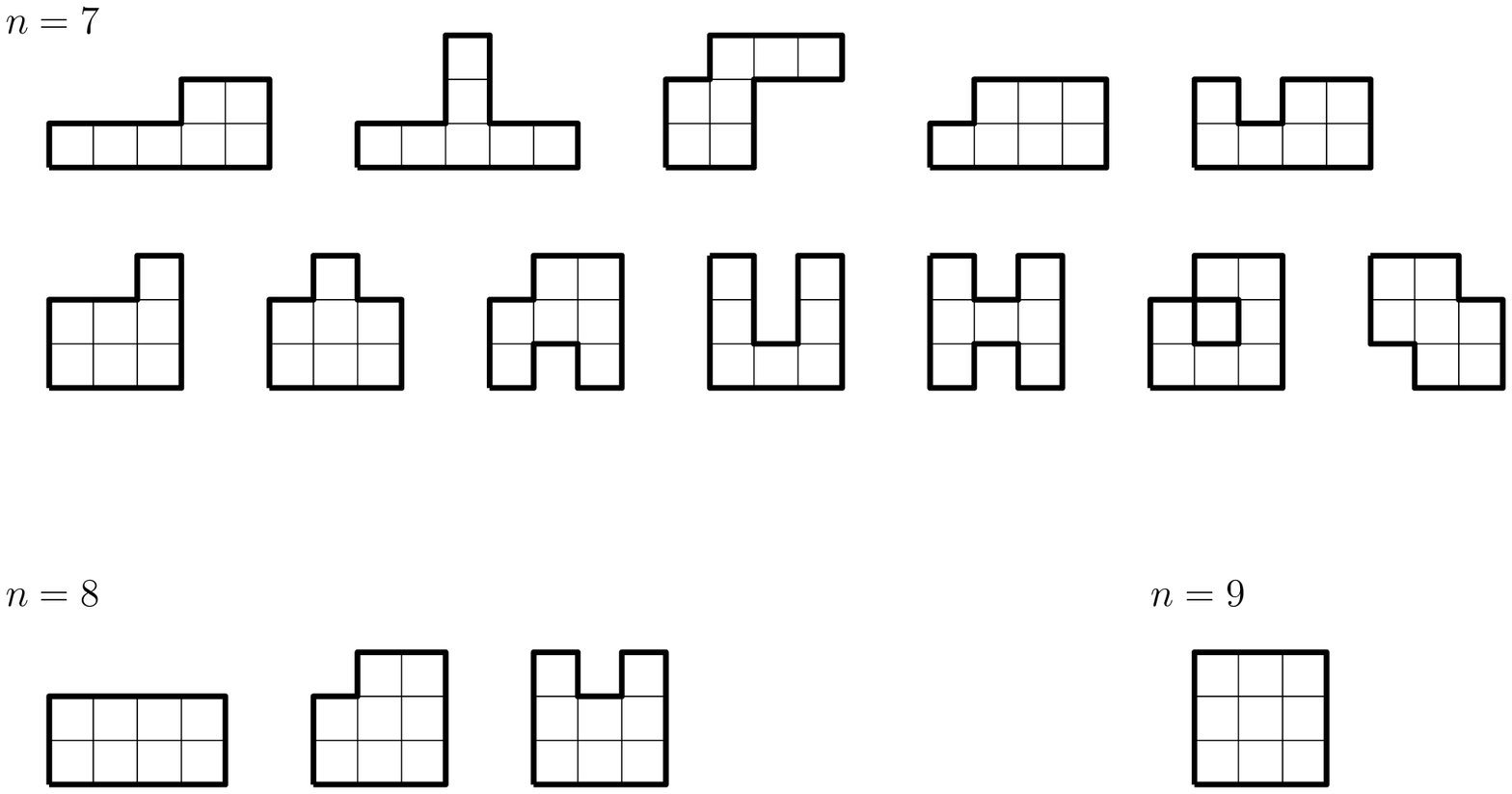}
\caption{Polyominoes of size 7 to 9 that have dual trees which cannot be folded to a cube.}
\label{fig:polyominoes070809nocube}
\end{figure}

Figures~\ref{fig:polyominoes06nocube}
and~\ref{fig:polyominoes070809nocube} show all polyominoes of size $n \geq
6$ for which dual trees exist (for $n=6$ for any existing dual tree), such that they cannot be folded to the
unit cube using $90^{\circ}$, $180^{\circ}$, and diagonal
folds. There are $24$ such polyominoes with a total of $43$ different dual trees for
$n=6$, $12$ polyominoes with $59$ dual trees for $n=7$, $3$ polyominoes with
$36$ dual trees for $n=8$, and one polyomino with $10$ dual trees for $n=9$.
Note, however, that for many of them there are cuts (i.e., dual trees)
such that they can be folded to the cube. For example, the $3\times 3$ square has $18$ dual trees that can fold
to a cube (Figure~\ref{fig:puzzle09}).

%
\begin{figure}[t]
\centering
\includegraphics[width=.7\columnwidth]{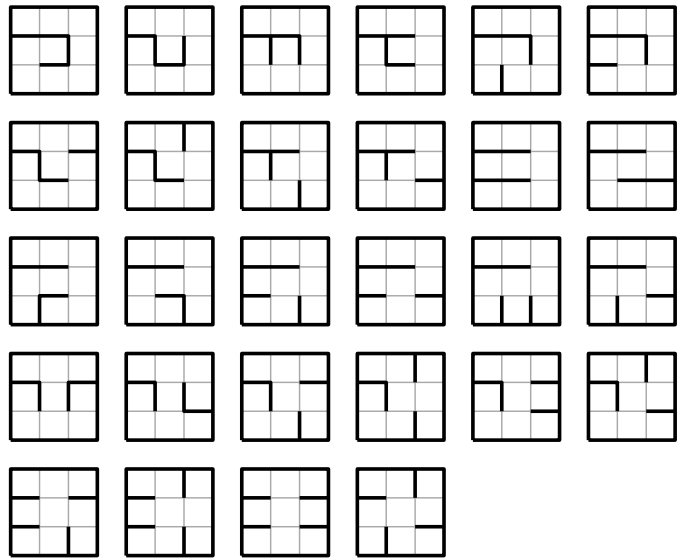}
\caption{\label{fig:puzzle09}
Puzzle\protect\footnotemark: which cuttings of a $3\times 3$ square
fold to a cube?
}
\end{figure}

{\interfootnotelinepenalty=10000
\footnotetext{Solution: The top row needs two $180^{\circ}$ folds,
and the second and third row need one $180^{\circ}$ fold. The last two rows cannot be
folded, even with diagonal folds. Note, however, that some of these cuttings
can be folded if we allow half-grid folds, cf. Figure~\ref{fig:lb9}.}
}

\section{Dynamic programming algorithm for trees}\label{sec:dp-tree}

\begin{theorem}
\label{thm:dp-tree}
Let $P$ be a tree shape, $Q$ be a polycube with $O(1)$ cubes, no four squares meeting at an edge, and $\mathcal F=\{ \pm 90^\circ\text{; grid}\}$. Then it is possible in linear time (in the size of $P$) to determine whether $P$ can fold to $Q$ in folding model $\mathcal F$.
\end{theorem}

\begin{proof}
We choose an arbitrary root of $P$. For each square $s$ of~$P$ define the subtree of~$s$ to be the tree shape consisting of all squares whose shortest path in~$P$ to the root passes through~$s$.
For a square $s$ of~$P$, define a \emph{placement} of~$s$ to be an identification of~$s$ with a surface square of~$Q$ together with the subset of the squares of~$Q$ covered by squares in the subtree of~$s$. We use a dynamic program that computes, for each square $s$ of~$P$, and each placement of~$s$, whether there is a folding of the subtree of $s$ that places $s$ in the correct position and correctly covers the specified subset.
Each square has $O(1)$ placements, and we can test whether a placement has a valid folding in constant time given the same information for the children of $s$. Therefore, the algorithm takes linear time.
\end{proof}

We have been unable to extend this result to folding models that allow $180^\circ$ folds, nor to folds with interior faces, nor to polycubes for which four or more squares meet at an edge. The difficulty is that the dynamic program constructs a mapping from the polyomino to the polycube surface (topologically, an \emph{immersion}) but what we actually want to construct is a three-dimensional embedding of the polyomino without self-intersections, and in general testing whether an immersion can be lifted to a three-dimensional embedding is  NP-complete~\cite{EppMum-SODA-09}.
For $90^\circ$ folds, a three-dimensional lifting always exists, as can be seen by induction on the number of squares in the tree shape: given a folding of all but one square of the tree shape, there can be nothing blocking the addition of the one remaining square to its neighbor in the tree shape.
However, in the presence of $+180^\circ$ folds, the addition of a square to an edge of a tree shape can be blocked by a $+180^\circ$ fold surrounding that edge.


It is tempting to attempt to extend our dynamic program to a
fixed-parameter algorithm for non-trees (parameterized by feedback vertex
number in the dual graph of the polyomino), 
by finding an approximate minimum
feedback vertex set, trying all placements of the squares in this set, and
using dynamic programming on the remaining tree components of the graph.
However, the problem of parts of the fold blocking other parts of the fold becomes even more severe in this case, even for $90^\circ$ folds. Additionally, we must avoid knots and twists in the three-dimensional embedding.  These issues make it difficult to extend the
dynamic program to the non-tree case.

\section{Triangular Grid}\label{sec:triang}

In this section we discuss folding \emph{polyiamonds} into \emph{polytetrahedra}. A \emph{polyiamond} is a 2D polygon that is a connected union of a collection of equilateral triangles on an equilateral triangular lattice. That is, polyiamonds are the analog to polyominoes with equilateral triangles instead of unit squares. We use the same notation as with polyominoes, so the number of triangles comprising a polyiamond is called the \emph{size} of the polyiamond, and we say a polyiamond is a \emph{tree-polyiamond} if its dual graph is a tree. A \emph{polytetrahedron} extends the notion of a \emph{polycube}, simply replacing unit cubes with regular tetrahedra. (Note, however, that the shapes of polytetrahedra are limited by the fact that regular tetrahedra do not tile space.)

We begin by noting that there are exactly $6$ polyiamonds that are tree shapes but do not fold into a tetrahedron.

\begin{lemma}\label{le:polyiamond}
In the folding model $\F = \{$any$^{\circ}$; $\Delta$grid$\}$, a tree-polyiamond folds into a tetrahedron if and only if it is not one of the following shapes:
\begin{center}
\includegraphics[width=0.75\textwidth]{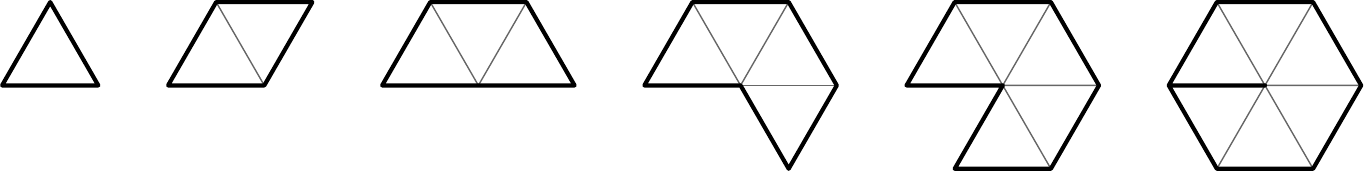}
\end{center}
\end{lemma}

\begin{proof}
The tetrahedron has two nets: one whose dual has a vertex of degree $3$, and one whose dual is a path of length $4$ and whose triangles do not all share a vertex. If a polyiamond contains either of the two nets, those four triangles may be folded into a tetrahedron with all of the other triangles of the polyiamond folded flat against the faces of the tetrahedron.

Therefore, we need to show that the above polyiamonds are the only polyiamonds that do not contain a net of a tetrahedron, and that none of the six can be folded into a tetrahedron in some other way. If the dual tree of the polyiamond contains a vertex of degree $3$, the polyiamond contains a net of the tetrahedron and can therefore be folded into a tetrahedron. Therefore, it is sufficient to consider polyiamonds whose dual trees are paths. If we examine a polyiamond whose dual is a path that does not contain the second tetrahedral net, either the polyiamond contains fewer than $4$ triangles, or each of its triangles share a single vertex. The only possibilities are the six mentioned above.
None can cover a tetrahedron, because they cannot reach the tetrahedron face that is opposite the shared vertex.
\end{proof}
\section{Conclusion}\label{sec:con}


Various open problems remain.

First, we gave an example of a tree shape $P$ that does fold into a polycube $Q$ for $\F = \{$any$^{\circ}$; interior faces; grid$\}$, but not in weaker models. In particular, $P$ has no folding into $Q$ that avoids using interior faces. $Q$ consists of 5 unit cubes---is it minimal?

Second, we characterized tree shapes that are a subset of a $3 \times n$ strip and fold into a unit cube in the $\F = \{$any$^{\circ}$; grid$\}$ model. Can we characterize polyominoes with holes (possibly of area zero) that fold into a unit cube, such as the original motivating puzzles in Figure~\ref{puzzles}, or our new puzzle in Figure~\ref{fig:puzzle2}?  What about polyominoes that are not a subset of a $3 \times n$ strip?

Third, if a tree shape folds into a unit cube, can the folding be performed while keeping the faces rigid (continuous blooming \cite{Blooming_GC})?
        
\begin{figure}[h]
\centering
    \includegraphics[width=1.3in]{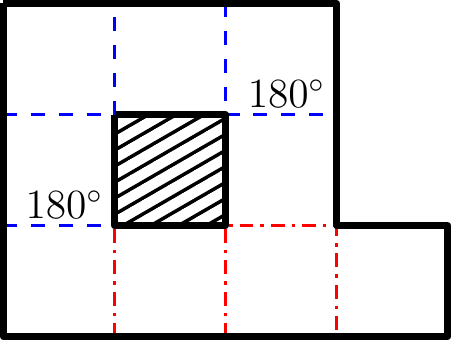}
         \caption{\label{fig:puzzle2} A polyomino with 9 squares and a hole that folds into a unit cube. Mountain (red) and valley folds (blue) are indicated, where two of the valley folds are $180^\circ$ folds.
         }
\end{figure}

\section*{Acknowledgments}
This research was performed in part at the 29th Bellairs Winter Workshop on Computational Geometry. We thank
all other participants for a fruitful atmosphere.
Oswin Aichholzer was partially supported by the ESF EUROCORES programme EuroGIGA --
CRP ComPoSe, Austrian Science Fund (FWF): I648-N18.
Irina Kostitsyna was supported in part by the NWO under project no. 612.001.106, and by F.R.S.-FNRS.
David Eppstein was supported by NSF Grants CCF-1228639, CCF-1618301, and CCF-1616248 and ONR/MURI Grant No. N00014-08-1-1015.

\small
\bibliographystyle{plain}
\bibliography{folding}

\end{document}